\documentclass[11pt]{article}
\pdfoutput=1
\usepackage{graphicx} 
\usepackage{sidecap}

\usepackage{fullpage}
\usepackage[margin = 1in]{geometry}

\usepackage{authblk}

\usepackage{amsmath, amssymb, amsthm,bbm, mathtools}
\usepackage{algorithmic}
\usepackage[ruled]{algorithm}
\usepackage{graphicx}
\usepackage{color}
\usepackage[dvipsnames]{xcolor}
\usepackage{enumerate}
\usepackage{thm-restate}

\usepackage{tcolorbox}

\usepackage[hidelinks]{hyperref}
\usepackage{cleveref}
\usepackage[round]{natbib}

\usepackage{nicefrac}

\newcommand{\opt}{\operatorname{OPT}}

\newcommand{\R}{\mathbb{R}}

\newcommand{\calF}{\mathcal{F}}

\newtheorem{theorem}{Theorem}[section]
\newtheorem*{theorem*}{Theorem}

\newtheorem*{proposition*}{Proposition}
\newtheorem{lemma}[theorem]{Lemma}
\newtheorem*{lemma*}{Lemma}
\newtheorem{corollary}[theorem]{Corollary}
\newtheorem*{conjecture*}{Conjecture}

\newtheorem*{fact*}{Fact}

\newtheorem*{hypothesis*}{Hypothesis}
\newtheorem{conjecture}[theorem]{Conjecture}

\newtheorem{claim}[theorem]{Claim}
\newtheorem*{claim*}{Claim}

\theoremstyle{definition}

\theoremstyle{remark}
\newtheorem{remark}[theorem]{Remark}
\newtheorem*{remark*}{Remark}

\newcommand{\eat}[1]{}

\newcommand{\N}{\mathbb{N}}

\newcommand{\Esymb}{\mathbb{E}}
\newcommand{\Psymb}{\mathbb{P}}

\DeclareMathOperator*{\E}{\Esymb}
 
 \DeclareMathOperator*{\ProbOp}{\Psymb}
 
 \DeclareMathOperator*{\argmin}{argmin}
 \DeclareMathOperator*{\argmax}{argmax}
\DeclareMathOperator*{\vol}{\mathrm{vol}}

\renewcommand{\Pr}{\ProbOp}

\newcommand{\tr}{\mathrm{tr}}

\renewcommand{\epsilon}{\varepsilon}

\newcommand{\Ball}{\mathrm{Ball}}

\newcommand{\chat}{\widehat{c}}
\newcommand{\Yhat}{\widehat{Y}}
\newcommand{\ellhat}{\widehat{\ell}}
\newcommand{\high}{\mathrm{high}}
\newcommand{\low}{\mathrm{low}}

\newcommand{\whM}{\widehat{M}}
\newcommand{\tauhat}{\widehat{\tau}}
\newcommand{\whE}{\widehat{E}}

\newif\ifnotes\notesfalse

\ifnotes
\usepackage{color}
\definecolor{mygrey}{gray}{0.50}
\newcommand{\notename}[2]{{\textcolor{blue}{\footnotesize{\bf (#1:} {#2}{\bf ) }}}}

\newcommand{\vnote}[1]{{\notename{Vaidehi}{#1}}}

\newcommand{\anote}[1]{{\notename{Aravindan}{#1}}}
\else

\newcommand{\notename}[2]{{}}

\newcommand{\enote}[1]{}
\newcommand{\vnote}[1]{}
\newcommand{\bnote}[1]{}
\newcommand{\anote}[1]{}

\fi

\newif\ifellipsoid\ellipsoidtrue

\ifellipsoid

\newcommand{\ellipsoidwriting}[1]{#1}
\else

\newcommand{\ellipsoidwriting}[1]{}

\fi

\usepackage{soul}

\newcommand{\email}[1]{\href{mailto:#1}{\texttt{#1}}}

\title{Computing High-dimensional Confidence Sets for Arbitrary Distributions}
\author{    
    Chao Gao\thanks{\email{chaogao@uchicago.edu}, Department of Statistics, University of Chicago, Chicago, USA} , 
    Liren Shan\thanks{\email{lirenshan@ttic.edu}, Toyota Technological Institute at Chicago, Chicago, USA} ,  Vaidehi Srinivas\thanks{\email{vaidehi@u.northwestern.edu}, Department of Computer Science, Northwestern University, Evanston, USA} ,  Aravindan Vijayaraghavan\thanks{\email{aravindv@northwestern.edu}, Department of Computer Science, Northwestern University, Evanston, USA}
    }

\date{}

\begin{document}

\maketitle

\begin{abstract}
    We study the problem of learning a high-density region of an arbitrary distribution over $\mathbb{R}^d$.  
    Given a target coverage parameter $\delta$, and sample access to an arbitrary distribution $\mathcal{D}$, we want to output a confidence set $S \subset \mathbb{R}^d$ such that $S$ achieves $\delta$ coverage of $\mathcal{D}$, i.e., $\mathbb{P}_{y \sim \mathcal{D}} \left[ y \in S \right] \ge \delta$, and the volume of $S$ is as small as possible. This is a central problem in high-dimensional statistics with applications in high-dimensional analogues of finding confidence intervals, uncertainty quantification, and support estimation. 
    
    In the most general setting, this problem is statistically intractable, so we restrict our attention to competing with sets from a concept class $\mathcal{C}$ with bounded VC-dimension. An algorithm for learning confidence sets is competitive with class $\mathcal{C}$ if, given samples from an arbitrary distribution $\mathcal{D}$, it outputs in polynomial time a set that achieves $\delta$ coverage of $\mathcal{D}$, and whose volume is competitive with the smallest set in $\mathcal{C}$ with the required coverage $\delta$. This problem is computationally challenging even in the basic setting when $\mathcal{C}$ is the set of all Euclidean balls.  Existing algorithms based on coresets find in polynomial time a ball whose volume is  $\exp(\tilde{O}( d/ \log d))$-factor competitive with the volume of the best ball.  

    
     Our main result is an algorithm that finds a confidence set whose volume is $\exp(\tilde{O}(d^{1/2}))$ factor competitive with the optimal ball having the desired coverage. It is surprisingly simple and also extends to finding confidence sets competitive against unions of $k$ balls, and improved guarantees under additional assumptions. 
     The algorithm is improper (it outputs an ellipsoid). Combined with our computational intractability result for proper learning balls within an $\exp(\tilde{O}(d^{1-o(1)}))$ approximation factor in volume, our results provide an interesting separation between proper and (improper) learning of confidence sets. 

\end{abstract}

\newpage
\tableofcontents
\newpage 

\section{Introduction}


\newcommand{\Bhat}{\widehat{B}}
\newcommand{\Rhat}{\widehat{R}}

We consider the problem of learning a high-density region of an arbitrary distribution.  That is, given an arbitrary distribution \(\mathcal{D}\) over \(\mathbb{R}^d\), and a coverage factor \(0 \le \delta \le 1\), we want to find the set \(S \subset \mathbb{R}^d \) that minimizes volume, while achieving coverage at least \(\delta\):
\begin{equation} \min \vol(S) \quad \text{s.t. } S \subset \mathbb{R}^d, \quad \mathbb{P}_{y \sim \mathcal{D}} [y \in S] \ge \delta, \label{eq:gqf}
\end{equation}
where $\vol(S)$ denotes the volume or Lebesgue measure of the set $S$. 
The minimum of (\ref{eq:gqf}) is known as the generalized quantile function \citep{einmahl1992generalized,polonik1997minimum}, an important quantity in statistical inference that plays a central role in a number of problems including estimating density level sets \citep{garcia2003level}, distribution support estimation \citep{scholkopf2001estimating} and conformal prediction \citep{ourwork2024}.
A set containing at least \(\delta\), say $0.9$, probability mass generalizes the idea of confidence intervals in one-dimension to a high-dimensional setting.

Solving (\ref{eq:gqf}) leads to a number of interesting applications in statistical inference.
On the one hand, finding a set that contains a large probability mass of the distribution can be useful as a high-dimensional analogue of support estimation \citep{scholkopf2001estimating}. On the other hand, as $\delta \to 0$, the maximizer of (\ref{eq:gqf}) can be used as a mode estimator ~\citep{sager1978estimation,sager1979iterative}. 
Algorithms that learn confidence sets for arbitrary distributions can also be used as a tool for estimating the uncertainty in the predictions of black-box machine learning models, as in conformal prediction \citep{gammerman1998learning,Vovkbook}.
Finally, algorithms for this task have to be robust to contamination present in the data, as the remaining $(1-\delta)$ fraction of the data could be outliers. The formulation in \eqref{eq:gqf} leads to a different notion of robustness that is non-parametric and may be of independent interest in robust estimation \citep{rousseeuw1984least,rousseeuw1985multivariate,van2009minimum}. There is also a rich body of work on robust statistics, where the remaining $1-\delta$ fraction can be arbitrary outliers, but these works often make stronger model assumptions about the $\delta$ fraction of points that are ``inliers'' ~\citep{huber1964robust, ars_book}. 
See \Cref{sec:applications} for applications in statistics like conformal prediction, goodness-of-fit, robust estimation and testing unimodality. 

\paragraph{Density estimation and curse of dimensionality.} A popular approach to learn minimum volume sets is based on the fact that the solution to (\ref{eq:gqf}) is given by level sets of the density function $\{y\in\mathbb{R}^d: p(y)\geq t\}$. The volume optimality of taking the level set of a density estimator can be established in some nonparametric settings 
\citep{rigollet2009optimal}. 
However, these methods typically assume additional properties of the density function of the distribution, and are usually limited to low dimensional settings since the sample complexity of nonparametric density estimation suffers from an exponential dependence on the dimension. 

\paragraph{Learning Confidence Sets with Bounded VC-dimension.} 
Consider the setting where there exists a good confidence set belonging to a family \(\mathcal{C}\) of finite VC-dimension. When independent samples from the distribution \(\mathcal{D}\) are available, one can solve the empirical version of (\ref{eq:gqf}) with polynomial sample complexity and establish approximate volume optimality as long as the concept class \(\mathcal{C}\) has finite VC-dimension \citep{polonik1999concentration,scott2005learning}. This is related to the (agnostic) PAC learning setting in classification where the confidence sets $S$ can be identified with classifiers belonging to the concept class $\mathcal{C}$.  This leads to the following broad class of learning problems that we study: 

\begin{tcolorbox}
\noindent \textbf{Learning confidence sets competitive with $\mathcal{C}$.} 
{\em An algorithm for learning confidence sets is $\Gamma$-competitive with $\mathcal{C}$, if given samples from an arbitrary distribution $\mathcal{D}$ in $\R^d$, for any $\gamma>0$, it finds a set $S \subset \mathbb{R}^d$ in time polynomial (in $d,1/\gamma$) satisfying $\mathbb{P}_{y \sim \mathcal{D}}[y \in S] \ge \delta$, and }
\begin{align}
   \vol(S)^{1/d} \le \Gamma \cdot \min\left\{ \vol(C)^{1/d}: C \in \mathcal{C} \text{ s.t. }\mathbb{P}_{y \sim \mathcal{D}} [y \in C] \ge \delta+\gamma\right\} \label{eq:goal-ball}.
\end{align}
\end{tcolorbox}
The small loss of the $\gamma$ term in coverage accounts for any potential sampling error. Above, we have followed the standard convention of measuring $\vol^{1/d}$ instead of volume $\vol$ (e.g., for Euclidean balls this becomes proportional to the radius of the ball 
). The factor $\Gamma$ corresponds to the approximation factor or competitive ratio of the algorithm. 

If there exists a small confidence set $C^* \in \mathcal{C}$ with (slightly more than) the required coverage $\delta$, then the algorithm will output a set $S$ with coverage $\delta$ and whose volume is competitive with $C^*$.
The set $S$ is not required to belong to the family $\mathcal{C}$. This is analogous to the standard setting in PAC learning, where the learning algorithm may output a classifier that does not belong to the concept class $\mathcal{C}$~\citep[see e.g.,][]{SSBD-book}. If the algorithm outputs a confidence set $S$ that belongs to $\mathcal{C}$, we will call it a {\em proper} learner of confidence sets (analogous to proper PAC learning). When we want to be more explicit, the term {\em improper} learner may be used for an algorithm that does not necessarily output a set belonging to the concept class $\mathcal{C}$. 
\anote{added:} We emphasize the worst-case nature of the distribution that the samples are drawn i.i.d from. This is analogous to the distribution-free setting of PAC learning, where we make no assumptions on the distribution.  
While there is a simple algorithm that is sample efficient (but computationally inefficient) when the VC-dimension of $\mathcal{C}$ is bounded (by a polynomial), the algorithmic problem becomes challenging in high dimensions.


\paragraph{The family of Euclidean Balls.}
Consider the basic setting where \(\mathcal{C}\) is the set of $\ell_2$ balls (this has VC-dimension at most $d$). \anote{added:} This is a natural generalization of intervals in higher dimensions, and is motivated by applications~\cite[see e.g.,][]{wang2023probabilistic}. 
It is NP-hard to find the smallest set in \(\mathcal{C}\) that covers a $\delta$ fraction of points from an arbitrary set \citep{ballsnphard}.\footnote{On the other hand, finding the minimum volume ball that encloses all the given $n$ points is polynomial time solvable.} The famous work of Badoiu, Har-Peled and Indyk on coresets for minimum volume enclosing balls~\citep{badoiu2002approximate, coreset-survey, ding:LIPIcs.ESA.2020.38} gives a proper learning algorithm that runs in polynomial time and is $\Gamma=1+ O(1/\log d)$-competitive \anote{should it be $1+O(\frac{(\log\log d)^2}{\log d})$ factor?} i.e., a $\exp\big(O(d/\log d)\big)$ factor approximation in volume. To the best of our knowledge, these works using core-sets \citep{badoiu2002approximate, ding:LIPIcs.ESA.2020.38} remain the state-of-the-art guarantee for learning confidence sets competitive with Euclidean balls in $d$ dimensions. \anote{Check this, esp. considering the more recent reference in review.} We focus on Euclidean balls and unions of Euclidean balls in the rest of the paper.    

\subsection{Our Results}

Our main algorithmic result is an algorithm that is $\Gamma \approx 1 + d^{-1/2}$ competitive against Euclidean balls.  


\vnote{thm-restate version of above}

\begin{restatable}[Learning confidence sets competitive with Euclidean balls]{theorem}{smallvolumeellipsoid}
    There is a polynomial time algorithm that for any target coverage \(\delta \in (0, 1)\) and coverage slack $\gamma\in(0,1)$, when given $n=\Omega(d^2/\gamma^2)$ samples drawn i.i.d. from an arbitrary distribution $\mathcal{D}$,  finds with high probability a set $S \subset \mathbb{R}^d$ that is $\Gamma= \exp \left(O_{\gamma, \delta} \big(d^{-1/2 + o(1)}\big)\right)$ competitive, i.e., 
    \[\mathbf{P}_{y \sim \mathcal{D}} \left[ y \in S \right] \ge \delta ,\]
    and 
    \[\vol(S)^{1/d} \le \vol(B^{\star})^{1/d} \left(1 + O \left(d^{-1/2 + o(1)} (\gamma \delta)^{-1} \right)\right).\]
    where \(B^\star\) is the minimum volume ball that achieves at least \(\delta + \gamma + O(\sqrt{d^2/n})\) coverage over \(\mathcal{D}\).\label{thm:intro:ellipsoid}
\end{restatable}


Please see \Cref{sec:ellipsoid} for more details. 
To compare against existing methods, consider the setting when $\delta, \gamma$ are constants. A straightforward solution that considers balls around the sample points that covers at least $\delta$ fraction of the data points achieves a factor of $\Gamma=2$ (due to the triangle inequality).\footnote{In fact this strategy achieves a $\Gamma=2$ factor approximation (with $\textrm{vol}^{1/d}$) for any set \(\mathcal{C}\) that consists of scalings and translations of a fixed convex shape. However, it is NP-hard to obtain a $(2-o(1))$ factor approximation even for the $\ell_\infty$ ball~\citep{ballsnphard}. } 
Algorithms based on the powerful technique of coresets for minimum volume enclosing balls due to ~\citet{badoiu2002approximate} gives $\Gamma=1+ \tilde{O}(1/\log d)$-competitive. 
In terms of the volume, the $\Gamma\approx 1 + d^{-1/2}$ factor corresponds to volume approximation of
\[\Gamma^d = \big(1+O_{\delta,\gamma}(d^{-1/2+o(1)})\big)^d = \exp \left(O_{\gamma, \delta} \big(d^{1/2+o(1)}\big)\right),\] 
which corresponds to a substantial improvement from the $\exp\big(O_{\delta, \gamma}(d/\log d)\big)$ volume approximation factor in existing work \citep{badoiu2002approximate, ding:LIPIcs.ESA.2020.38}. 

\paragraph{Union of Balls $\mathcal{C}_k$.} We next consider the more general concept class $\mathcal{C}_k$ which is the unions of $k$ Euclidean balls. This family has VC-dimension of at most $kd$. We obtain a similar guarantee for unions of $k=O(1)$ balls, by recursively applying the algorithm from \Cref{thm:intro:ellipsoid}.     

\begin{theorem}[Union of Balls] \label{thm:intro:improper:unions}
 Let $\delta \in (0,1), \gamma \in (0,1), k \in \mathbb{N}$ be any constants. There is a polynomial time algorithm that for target coverage \(\delta \in (0, 1)\) and coverage slack $\gamma\in(0,1)$ when given $n=\Omega(kd^2/\gamma^2)$ samples drawn i.i.d. from an arbitrary distribution $\mathcal{D}$,  finds with high probability a set $S \subset \mathbb{R}^d$ that is $\Gamma=\left(1 + O_{\gamma, \delta} \big(d^{-1/2 + o(1)}\big)\right)$ competitive; more precisely, it satisfies
    $\mathbf{P}_{y \sim \mathcal{D}} \left[ y \in S \right] \ge \delta$,
    and 
    \[\vol(S)^{1/d} \le  \mathrm{vol}(C^\star_k)^{1/d} \left(1+ O_{k,\delta} \left(d^{-1/2 + o(1)} \right)\right) \cdot \left( \frac{O(\log(k/\gamma))}{\gamma} \right)^{1/d}  \]
    where \(C^\star_k\) is the minimum volume union of $k$-balls that achieves at least \(\delta + \gamma + O(\sqrt{kd^2/n})\) coverage over \(\mathcal{D}\).
\end{theorem}
See  \Cref{cor:union-of-ellipsoids} for details. Our algorithms are surprisingly simple and conceptually different from existing algorithms based on coresets. It is based on a new connection to robust estimation that follows from a structural statement about the mean and the  center of a large fraction of points in the optimal ball with the desired coverage.  
We give a detailed description of the ideas in the technical overview.



\subsection{Proper Learning: Separation Result and Algorithms}

The algorithm in \Cref{thm:intro:ellipsoid}  outputs a set $S \subset \mathbb{R}^d$ of small volume achieving the desired coverage, when there is a small ball achieving the required coverage. The set $S$ may not be a ball, and may in general be an ellipsoid.  Similarly the algorithm in \Cref{thm:intro:improper:unions} outputs a union of $O(k/\gamma)$ ellipsoids. These correspond to ``improper'' learners for confidence sets. In contrast, the best {\em proper} learning algorithms achieve a worse competitive ratio of $\Gamma=1 + O(1/\log d)$. The following theorem shows that such a loss is unavoidable assuming $P \ne NP$, unless we are willing to make additional assumptions on $\mathcal{D}$.

\begin{theorem}[NP-hardness of Proper Learning for Balls]\label{thm:intro:nphard}
   For any constant $\varepsilon>0$, assuming $P \ne NP$ there is no polynomial time learning algorithm that given 
   samples from an arbitrary distribution over $\mathbb{R}^d$ that outputs balls as confidence sets that are $\Gamma \le (1+\frac{1}{d^{\varepsilon}})$ competitive with balls. 
\end{theorem}
See \Cref{thm:properhard} for a more formal statement. Note that the above result is a worst-case hardness result, and reflects the arbitrary nature of the distribution. This is reminiscent of NP-hardness of proper learning in supervised learning~\citep[see e.g.,][]{FGKP2009}. The above theorem follows from a simple reduction that amplifies the existing NP-hardness result due to \citep{ballsnphard}. In \Cref{thm:properhard:slack}, we also show computational intractability (assuming the Small Set Expansion hypothesis of \cite{raghavendra2010graph}) that shows a similar quantitative hardness, even when the coverage can be smaller by an arbitrary constant factor. Combined with our algorithmic result in \Cref{thm:intro:ellipsoid}, this shows a formal 
separation 
between proper and improper learning for a natural geometric learning problem. 

\paragraph{Proper learning algorithms.} When we are required to output a confidence set from $\mathcal{C}$, the algorithm in \Cref{thm:intro:ellipsoid} can also be used to output a ball with a worse competitive ratio (which is expected due to \Cref{thm:intro:nphard}). When we make no assumptions about the distribution $\mathcal{D}$, our algorithm can output a ball that is $\Gamma=1+\widetilde{O}(1/\log d)$ competitive. See \Cref{thm:main:ball} in \Cref{sec:balls} for a formal statement of the theorem. This guarantee recovers the guarantee from the coreset-based algorithm which achieves $\Gamma=1+\widetilde{O}(1/\log d)$ \citep{badoiu2002approximate}.  In the worst-case instances of the algorithm (and the hard instances in \Cref{thm:intro:nphard}),  the points sampled from $\mathcal{D}$ that lie inside the optimal ball $B^\star$ effectively lie on a lower-dimensional subspace. \vnote{updated phrasing here to not say we are slightly worse}

On the other hand, our algorithm can give a significantly stronger guarantee when the points inside $B^\star$ are approximately isotropic. In the following theorem, we consider the setting where we have $n$ points $Y$ drawn from $\mathcal{D}$, and the guarantee depends on the covariance of the sampled points within the optimal ball $B^\star$. 
\vnote{Adding reference here }

\begin{restatable}[Bounded variance implies better bounds]{theorem}{boundedvarianceimpliesbetterbounds}
    Let \(Y \subseteq \mathbb{R}^d\) be a set of \(n\) points, and \(\delta, \gamma \in (0,1)\), such that there is an unknown ball \(B^\star = B(c^\star, R^\star)\) with \(Y^\star = B^\star \cap Y\) satisfying \(|Y^\star| \ge \delta |Y|\), and \(Y^\star\) is \(\beta (R^\star)^2/d\)-isotropic for some $\beta \ge 0$.  That is, if \(\mu^\star\) is the mean of the points \(Y^\star\), then 
    \[\Sigma_{Y^\star} = \frac{1}{|Y^\star|} \sum_{y \in Y^\star} (y - \mu^\star)(y - \mu^\star)^\top \preccurlyeq \beta \frac{(R^\star)^2}{d} I.\]
    Then we can find a ball \(\Bhat = B(\chat, \Rhat)\) such that \(|Y \cap \Bhat| \ge (1 - \gamma) |Y^\star|\), and 
    \[\vol(\Bhat)^{1/d} \le \vol(B^{\star})^{1/d} \sqrt{1 +  O \left( \frac{\beta}{\gamma \delta d} \right)},\]
    in polynomial time. 
    \label{cor:ball-bounded-variance-better-bounds}
\end{restatable}

In the above \Cref{cor:ball-bounded-variance-better-bounds}, $\Sigma_{Y^\star}$ is a PSD matrix of dimension $d \times d$  with $\tr(\Sigma_{Y^\star}) \le (R^\star)^2$, since \(Y^\star \subseteq B^\star\). The isotropicity condition requires the maximum eigenvalue $\lambda_{\max}(\Sigma_{Y^\star}) \le \beta (R^\star)^2/d$ for some $\beta>0$.  Thus every \(Y^\star\) satisfies the isotropic condition with $\beta \le d$.  When the points of $Y^\star$ are very spread out, this corresponds to a setting when $\beta \le 1$ (all the eigenvalues are equal). In this case, the theorem guarantees to find a ball with $\Gamma= 1+O(d^{-1})$. In other words, it finds a ball whose volume is within $O_{\delta,\gamma}(1)$ factor  of the optimum, which is significantly stronger. Finally, the ball $\widehat{B}$ output by the algorithm also has coverage $(1-\gamma) \delta - O(\sqrt{d/n})$ by standard concentration arguments \citep[see e.g.,][]{devroye2001combinatorial} since the VC-dimension of $d$-dimensional Euclidean balls is $d+1$.

All of our algorithmic results are distribution-free i.e., they do not make any assumptions on the distribution $\mathcal{D}$. While it may be impossible to learn the distribution or even estimate the density at a given point (these tasks often incur an exponential dependence on the dimension $d$ even for smooth distributions), our results show that one can obtain polynomial time algorithms that find approximately optimal dense sets from $\mathcal{C}$ covering at least $\delta$ mass. 

\subsection{Application to Conformal Prediction}

As a direct application of our result we get the following algorithm for conformal prediction in high dimensions.  In conformal prediction, an algorithm is given training examples \(Y_1, \dots, Y_n\) in some set \(\mathcal{Y}\), and a miscoverage rate \(\alpha\), and must output a set \(\widehat{C}\), where 
\[\Pr[Y_{n + 1} \in \widehat{C}] \ge 1 - \alpha,\]
for some unseen training example \(Y_{n + 1}\), assuming only that the training and text examples \(Y_1, \dots, Y_{n + 1}\) are \emph{exchangeable}.\footnote{That is, \(\Pr[Y_1 = y_1, \dots, Y_{n + 1} = y_{n + 1}] = \Pr[Y_1 = y_{\pi(1)}, \dots, Y_{n + 1} = y_{\pi(n + 1)}]\), for all \(y_1, \dots, y_{n + 1} \in \mathcal{Y}\) and permutations \(\pi: [n + 1] \rightarrow [n + 1]\).}  There is a fairly simple way to ``conformalize" our result to achieve this goal, resulting in the following algorithm that always outputs a valid conformal set when the data is exchangeable, and is additionally approximately volume optimal when the data is drawn i.i.d.\ from an (unknown) distribution \(\mathcal{D}\).

\begin{restatable}[Conformal Prediction with Approximate Volume Optimality]{theorem}{highdimensionalconformalprediction}\label{thm:high-dim-conformal-prediction}
We have an algorithm for conformal prediction over examples from \(\mathcal{Y} = \mathbb{R}^d\), that achieves \emph{approximate volume optimality} with respect to the set \(\mathcal{C}\) of Euclidean balls.\footnote{This is an approximate form of {\em restricted volume optimality} defined by \cite{ourwork2024}.}  That is, we have an algorithm that, given training examples \(Y_{1}, \dots, Y_n \in \mathcal{Y}\), miscoverage rate \(0 < \alpha < 1\), and coverage slack factor \(0 \le \gamma \le 1\), outputs a set \(\widehat{C}\) (not necessarily in \(\mathcal{C}\)), such that for an unknown test example \(Y_{n + 1} \in \mathcal{Y}\),
    \begin{enumerate}[(a)]
        \item if \(Y_1, \dots, Y_{n + 1}\) are exchangeable, then 
        \[\Pr [Y_{n + 1} \in \widehat{C} \ge 1 - \alpha].\]
        \item if \(Y_1, \dots, Y_{n + 1}\) are drawn i.i.d.\ from some (unknown) distribution \(\mathcal{D}\), and \(n =  \Omega(d^2/\gamma^2)\), then 
        \[\mathrm{vol}(\widehat{C})^{1/d} \le \left(1 + O_{\gamma, \delta} \big(d^{-1/2 + o(1)}\big)\right) \mathrm{vol}(C^\star)^{1/d},\]
        where 
        \[C^\star = \argmin_{C \in \mathcal{C}} \mathrm{vol}(C) \qquad \text{s.t.} \quad \Pr[Y_{n + 1} \in C] \ge 1 - \alpha + \gamma.\]
    \end{enumerate} 
\end{restatable}

We provide an in depth discussion of this application in \Cref{sec:conformal}. See also \Cref{sec:applications} for more applications in high-dimensional statistics. We now review related literature and other methods that have been explored in literature, followed by an overview of the algorithm and techniques.

\subsection{Prior and Related Algorithmic Approaches}
\label{sec:priorwork}

The problem of learning minimum volume sets has a long history in statistical learning and computer science. It also goes under the names of learning density level sets, anomaly detection, or distribution support estimation in the literature. Learning algorithms that have been studied in the literature can roughly be categorized into the following three classes.
\begin{enumerate}
\item \textit{Excess mass functional.} The empirical version of \Cref{eq:gqf} is
\begin{equation}
\min\vol(C)\quad\text{s.t. }C\in \mathcal{C}\text{ and }\mathbb{P}_n(C)\geq \delta -\gamma, \label{eq:ERM}
\end{equation}
where $\mathbb{P}_n$ is the empirical distribution and $\gamma$ is some slack parameter \citep{scott2005learning}. A more general objective function called the excess mass functional proposed by \cite{hartigan1987estimation} is defined by
\begin{equation}
E_{\lambda}(C) = \mathbb{P}_n(C) - \lambda\vol(C), \label{eq:EMF}
\end{equation}
where the parameter $\lambda>0$ is determined by the confidence level. In fact, maximizing the excess mass functional over $C\in\mathcal{C}$ is equivalent to \Cref{eq:ERM}. This approach has been well studied in the literature. The original work \citep{hartigan1987estimation} considers $\mathcal{C}$ being the collection of convex sets. The collection of ellipsoids is considered by \cite{nolan1991excess}. Convergence rates of the maximizer of \Cref{eq:EMF} are investigated by \cite{tsybakov1997nonparametric} in a nonparametric setting for star-shaped sets and convex sets.

\item \textit{One-class classification.} Another type of algorithms aim to learn a nonparametric function $f:\mathbb{R}^d\rightarrow\mathbb{R}$, and the the confidence set is taken as $\{y\in\mathbb{R}^d:f(y)\geq 0\}$. This approach was pioneered by \cite{scholkopf2001estimating} and they model the function $f$ in a feature space by
\begin{equation}
f(y)=\text{sign}\left(w^T\Phi(y)-\rho\right),
\end{equation}
where $\Phi:\mathbb{R}^d\rightarrow\mathbb{R}^k$ is some feature mapping. The parameters $w$ and $\rho$ are solutions to the following optimization problem,
\begin{equation}
\max_{w,\rho,\xi}\frac{1}{2}\|w\|^2+\frac{1}{\lambda n}\sum_{i=1}^n\xi_i-\rho\quad\text{s.t. }w^T\Phi(Y_i)-\rho+\xi_i\geq 0,\label{eq:OSVM}
\end{equation}
where $\xi\in\mathbb{R}^n$ is a margin vector and $\lambda$ is determined by the confidence level. The optimization (\ref{eq:OSVM}) is known as the one-class support vector machine (SVM) in the literature, which is equivalent to another optimization problem called support vector data description (SVDD) \citep{tax2004support} using some particular feature mapping. Recasting learning minimum volume sets as solving classification is natural. The excess mass functional (\ref{eq:EMF}) already suggests that the problem is equivalent to testing between two measures, empirical distribution against volume (uniform). Formally, this connection was established by \cite{steinwart2005classification}. Unlike traditional classification, here all the data points are generated from the same distribution, which leads to the name ``one-class classification" in the literature. Compared with directly learning confidence sets through (\ref{eq:ERM}) or (\ref{eq:EMF}), the one-class classification perspective is more flexible, and can be easily combined with deep learning models \citep{ruff2018deep}.

\item \textit{Density estimation plug-in.} The population version of the excess mass functional is $P(C)-\lambda\vol(C)$, which is maximized by the density level set
$$\{y\in\mathbb{R}^d: p(y)\geq t\},$$
by Neyman--Pearson lemma, since the density function $p$ can be regarded as the likelihood ratio between $P$ and volume \citep{garcia2003level}. Thus, it is natural to replace $p$ by some density estimator $\widehat{p}$. The plug-in strategy has been considered by \cite{hyndman1996computing,park2010computable,Lei2013DistributionFreePS} among others. Under certain nonparametric setting, the optimality of the plug-in strategy can be proved \citep{rigollet2009optimal}.
\end{enumerate}

\paragraph{Algorithms based on Coresets.} Coresets are a powerful algorithmic primitive for geometric problems involving an input set of data points~\citep{coreset-survey}. A coreset for a problem is a small summation of the data (e.g., a subset of the points) such that solving the problem on the summation gives an approximate solution to the whole instance. Badiou, Har-Peled and Indyk proved a surprising result about coresets for the problem of minimum volume enclosing ball for a set of points -- they showed that there is a coreset of size $1/\varepsilon$ (that is independent of the dimension) whose minimum volume enclosing ball approximates the minimum volume ball enclosing all the points up to a $(1+\varepsilon)$ factor in radius~\citep{badoiu2002approximate}.\footnote{The coreset size of $\lceil 1/\varepsilon \rceil$ is also known to be tight~\citep{coreset-survey}. } \citet{badoiu2002approximate} show how to find a ball enclosing $\delta - \gamma$ fraction of the given points in polynomial time whose radius is within a factor $\Gamma=\Big(1+ O_\gamma(\frac{(\log \log d)^2}{\log d})\Big)$
\citep[see Theorem 4.2 of][]{badoiu2002approximate}. The subsequent work of \citet{ding:LIPIcs.ESA.2020.38} also uses coresets and seems to obtain a guarantee of $\Gamma=\Big(1+ O_\gamma(\frac{1}{\log d})\Big)$  (see (24) following Corollary 11 in \citep{ding:LIPIcs.ESA.2020.38}), though this is not stated explicitly. 
This corresponds to an approximation factor of $\Gamma^d=\exp\big(\widetilde{O}(d/\log d)\big)$ for the volume. We are not aware of any other work that get a better competitive ratio $\Gamma$ for balls. 


\paragraph{Comparison to robust mean estimation and covariance estimation.} 
There is a rich literature for developing robust algorithms to perform basic statistical estimation tasks in high-dimensions when $1-\delta$ fraction of the points are contaminated~\citep[see e.g.,][]{huber2004robust,ars_book}. This is usually studied in the context of parameter estimation problems like mean estimation or covariance estimation 
 where a $\delta$ fraction of points are  {\em inliers} that are assumed to be drawn from a distribution with a location (mean) or covariance, while the rest of the points are arbitrary {\em outliers}. The inlier distribution is assumed to satisfy some nice property like bounded covariance and/or moments, or typically from a parameteric family like Gaussians. Various contamination models have been considered. In the famous Huber contamination model and the stronger adversarial contamination model, a $\delta=1-\eta > 1/2$ fraction of the points are {inliers}, while a $\eta=1-\delta < 1/2$ fraction of the points are outliers~\citep{huber1964robust, ars_book}. 
 The setting when $\delta < 1/2$ has also been studied as  list-decodable robust estimation~\citep{charikar2017list}.    
 Most relevant to us are the works in algorithmic robust statistics including the important works of \citet{diakonikolas2016robust,lai2016agnostic}. 
 which gave polynomial time algorithms for robust mean and covariance estimation in high dimensions. 
 Over the past decade there have been several polynomial time algorithms proposed based on outlier removal, sophisticated convex relaxations including sum-of-squares and many other techniques~\citep[please see][for an excellent book on the topic]{ars_book}. 
 
 However, the robust estimation literature usually assumes various distributional assumptions -- for mean estimation, they assume at the very least that the inlier distribution has bounded variance in every direction; while for covariance estimation one needs various fourth moment conditions, or other distributional assumptions like Gaussianity. Such assumptions are necessary to define a certifiable property of the inlier distribution in parameter estimation tasks. 
 
 In our setting, our main task is not one of parameter estimation. More importantly, the distribution of points is completely arbitrary i.e., we cannot make any assumptions on the distribution of both the inliers and the outliers. But in some settings we use powerful algorithms for list-decodable mean estimation~\citep{charikar2017list} as a black-box subroutine once we reduce the task to mean estimation for bounded covariance distributions. Finally, algorithms for our problem can also be used to design robust estimators for various statistical tasks like mode estimation or robust support estimation, as detailed in \Cref{sec:applications}.

\subsection{Applications in Statistics}\label{sec:applications}


Successfully learning a small confidence set has implications in many statistical inference tasks, ranging from estimation to hypothesis testing. 
We will summarize a few important applications below.
\begin{enumerate}
\item \textit{Conformal prediction.} The goal of conformal prediction \citep{gammerman1998learning,Vovkbook,angelopoulos2023survey} is to find a set $\widehat{C}$ that is trained from $y_1,\cdots,y_n\sim \mathcal{D}$ such that $\mathbb{P}(y_{n+1}\in \widehat{C})\geq 1-\alpha$ for an independent future observation $y_{n+1}\sim \mathcal{D}$. 

Minimum volume sets are natural candidates for conformal prediction to achieve finite sample coverage and volume optimality simultaneously. For example, conformalized density level sets have already been considered by \cite{Lei2013DistributionFreePS}. More generally, any minimum volume set learning algorithm, including our algorithms in the paper (e.g., Theorem~\ref{thm:intro:ellipsoid},  Theorem~\ref{thm:intro:improper:unions}), can be conformalized through the construction of nested systems \citep{gupta2022nested,ourwork2024} to obtain similar approximate volume optimality guarantees.    For more details see \Cref{sec:conformal}.

\item \textit{Robust estimation.} Suppose $\widehat{C}$ is a $\delta$-level confidence set with some $\delta$ close to zero. It was suggested by \cite{sager1978estimation,sager1979iterative} that one can use some point in $\widehat{C}$ as a mode estimator. More generally, robust location and scatter estimation through learning minimum volume sets is advocated by \cite{rousseeuw1984least,rousseeuw1985multivariate}. Specifically, Rousseeuw first computes $\widehat{C}$ as the minimum volume ellipsoid (MVE), and then the center and the shape of the $\widehat{C}$ are used as a robust location estimator and a robust scatter matrix estimator. Other applications of MVE are discussed in \cite{van2009minimum}.

\item \textit{Testing unimodality.} The original motivation of \cite{hartigan1987estimation} in learning a minimum volume convex set is to test whether a distribution is unimodal. The paper suggested testing statistic $\sup_{\lambda}\sup_{C\in\mathcal{C}}E_{\lambda}(C)$ that maximizes the excess mass functional (\ref{eq:EMF}). The idea was further developed and generalized by \cite{muller1991excess,polonik1995measuring,cheng1998calibrating}.

\item \textit{Anomaly detection.} The problem of anomaly detection, also known as outlier or novelty detection, aims to find a set that is not typical. Given a distribution $\mathcal{D}$ with density $p$, its anomaly region is $\{y\in\mathbb{R}^d:p(y)<t\}$, the complement of a density level set. Thus, the problem is mathematically equivalent to learning a confidence set. We refer the readers to \cite{chandola2009anomaly,ruff2021unifying} for comprehensive review of this area.

\item \textit{Goodness-of-fit.} Consider a null hypothesis $H_0:P=P_0$. Whether data is generated from $P_0$ can be determined by the value of $|\mathbb{P}_n(C)-P_0(C)|$ for some set $C$ that is informative. The idea of using minimum volume sets in the goodness-of-fit test was considered by \cite{polonik1999concentration}. In particular, let $\widehat{C}_{\alpha}$ (resp. $C_{\alpha}$) be a $(1-\alpha)$-level confidence set computed from $\mathbb{P}_n$ (resp. $P_0$). The statistic
$$\sup_{\alpha}\left(|\mathbb{P}_n(\widehat{C}_{\alpha})-P_0(\widehat{C}_{\alpha})|+|\mathbb{P}_n(C_{\alpha})-P_0(C_{\alpha})|\right)$$
was proposed and its asymptotic property was analyzed by \cite{polonik1999concentration}.

\end{enumerate}


\section{Technical Overview}

Our goal is, given a set \(Y\) of samples from \(\mathcal{D}\), to find a set \(S\) that captures \((1 - \gamma) \delta |Y|\) points of \(Y\), and has volume comparable to the minimum volume ball \(B^*\) that captures a subset of points \(Y^\star \subseteq Y\), with  \(|Y^\star| \ge \delta |Y|\).  If we could estimate the center of \(B^\star\) well, then we could find \(B^\star\) approximately by guessing an appropriate radius.  We can focus on the worst-case problem over the samples \(Y\) because we will choose our output set \(S\) to be an ellipsoid.  Since the class of ellipsoids in \(d\) dimensions has bounded VC-dimension, the empirical coverage of all ellipsoids enjoys uniform convergence.  This means that the coverage of \emph{all} ellipsoids will simultaneously generalize from the empirical samples to the population setting, so it suffices to find an ellipsoid that approximately minimizes volume subject to a coverage constraint on the empirical samples.

Intuitively, one could interpret our problem as trying to recover the \emph{center} of the points \(Y^\star \subseteq Y, |Y^\star| \ge \delta |Y|\) in a way that is robust to adversarial contamination, i.e., we only have access to \(Y\).  However, the center of the ball could depend on only a few points, i.e., a sphere in \(d\) dimensions can be defined by \(d + 1\) points.\footnote{In fact the coreset is a subset of size only $\lceil 1/\varepsilon \rceil$ that determines a ball that is within a $(1+\varepsilon)$ approximation in radius.} Thus, it appears that the center is not a quantity that can be robust to outliers.  This is in contrast to other quantities, such as the mean or median, that inherently depend on the whole dataset, and are not largely influenced by a few points. 

In fact, the goal of recovering the center of \(B^\star\) is too strong.  Properly learning the minimum volume ball \(B^\star\) within an approximation factor \(\Gamma = 1 + 1/d^\varepsilon\), for any constant \(\varepsilon > 0\) is NP-hard.  So we should not expect such a strategy to work in the worst case.  

However, it turns out that for certain non-worst-case instances, it is indeed possible to recover \(B^\star\) within a better approximation ratio than in the worst case.  Our first main technical insight is that when the variance of the points is well-spread (i.e., the sampled points do not have too much of their total variance concentrated in any one direction), the \emph{mean} of the points serves as a good proxy for the center.
As a side result, this implies strong proper learning guarantees for non-worst-case instances, where the approximation ratio improves as the variance is more controlled (see \Cref{thm:main:ball} and \Cref{cor:ball-bounded-variance-better-bounds} in \Cref{sec:balls}).    

This assumption on the variance of the sampled points is quite strong, and we cannot expect it to hold for arbitrary distributions.  Our second main technical insight is that it is possible to precondition the sampled points by applying a linear transformation \(T\) to the points, to move them into a position that meets the strong variance criterion above.  Then, we can solve the non-worst-case proper learning task in the transformed space, and apply the inverse transformation \(T^{-1}\) to the resulting ball to construct a confidence set for the original points.  

The transformation \(T\) necessarily distorts the space, so when we apply the inverse transformation \(T^{-1}\) to a ball there are two things to keep in mind.
First, the result will no longer be a ball, rather it will be an ellipsoid.  This falls in the regime of improper learning, and allows us to sidestep the strong lower bounds against proper learning.  Second, the transformation \(T\) ``shrinks" the space, so the inverse transformation \(T^{-1}\) will expand the volume of the ball.

Our final result follows from carefully choosing the transformation \(T\) to balance two competing objectives: 
\begin{enumerate}[(a)]
    \item We must control the variance of the transformed points to minimize the approximation ratio of the non-worst-case proper learning algorithm.
    \item We must control the distortion of \(T\) to minimize the volume blow up from applying \(T^{-1}\) to the result of the proper learning algorithm.
\end{enumerate}

Armed with these insights our final algorithm is quite simple: for each coarse set of candidate points (points contained in a ball centered at one of the sample points, there are \(O(n^2)\) such sets), we estimate a linear transformation \(T\), we transform the points by \(T\), we find the smallest ball centered at the mean of the transformed points that achieves the desired coverage, and we apply \(T^{-1}\) to the resulting ball to get a candidate ellipsoid.  Finally we take the minimum volume ellipsoid that is found over all coarse sets of candidates.

\subsection{Proper Learning for Non-Worst-Case Instances}

To tackle (a) above, our main observation is that we can approximate the center \(c^\star\) of the minimum volume ball \(B^\star\) containing \(Y^\star\) by using the mean \(\mu^\star\) of \(Y^\star\).  It is not in general true that \(\mu^\star\) should be near \(c^\star\).  However, \emph{most} points in \(Y^\star\) must be approximately as close to \(\mu^\star\) as they are to \(c^\star\).  
This is illustrated in \Cref{fig:ball:overview}.  

\begin{figure}[h]
\centering
    \centering
    \includegraphics[width=0.5\linewidth]{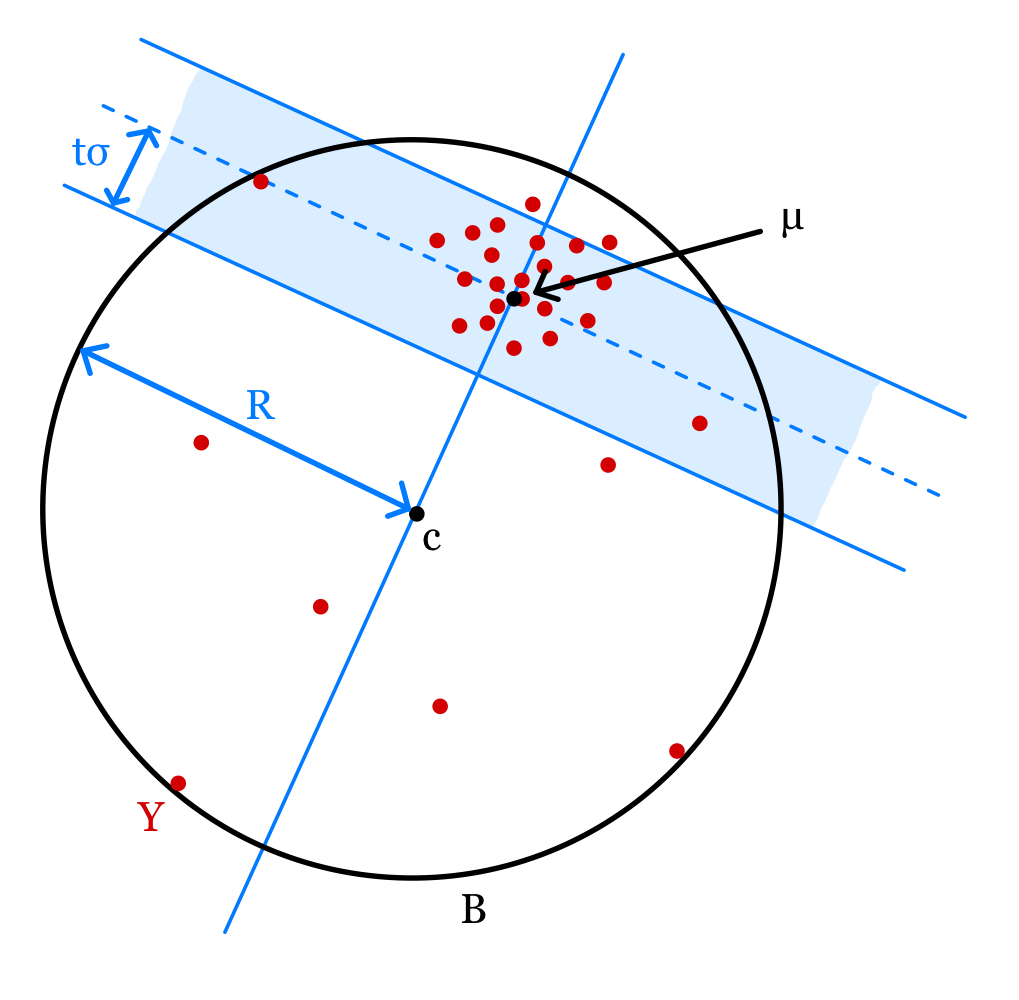}
    \caption{The figure shows the points \(Y\) (in red) with mean \(\mu\), that are contained in ball \(B\) with center \(c\) and radius \(R\).  It is not necessarily the case that \(\mu\) is near \(c\), as \(B\) can be defined by just a few points.  However, Chebyshev's inequality tells us that \emph{most} of the points \(Y\) (depicted here as the points in the shaded region) are within a few standard deviations (\(\sigma\)) of \(\mu\) in the \(\mu - c\) direction.  This allows us to bound the distance of these points to \(\mu\) as \(\le \sqrt{R^2 + (t\sigma)^2}\).}
    \label{fig:ball:overview}
\end{figure}

\begin{restatable}[For bounded variance, most points are near the mean]{lemma}{boundedvariancemostpointsnearmean}
    Let \(Y\) be a set of points in \(\mathbb{R}^d\) that are contained in a ball \(B(c, R)\).  Let \(\mu_Y\) be the empirical mean of \(Y\), and \(\Sigma_Y\) be the empirical covariance, where the largest eigenvalue of \(\Sigma_Y\) satisifies \(\lambda_{\mathrm{max}}(\Sigma_Y) \le \sigma^2\).  Let \(\widehat{\mu}\) be an arbitrary point such that \(||\mu_Y - \widehat{\mu}||_2 \le \tau\).  Then for any \(t \ge 1\), there exists a subset \(\widehat{Y}\) such that \(|\widehat{Y}| \ge \left(1 - \frac{1}{t^2} \right) |Y|\), and for every point \(y \in \widehat{Y}\),
    \[||y - \widehat{\mu}||^2_2 \le ||y - c||_2^2 + t^2 (\sigma^2 + \tau^2).\]
    \label{lem:ball-bounded-variance-pts-near-mean}
\end{restatable}

\begin{proof}
    As a warm-up, we can bound the number of points that deviate significantly from \(\mu_Y\) using Chebyshev's inequality.  Let \(v_\mu = \frac{\mu_Y - c}{||\mu_Y - c||}\) be the unit vector in the direction of \(\mu_Y - c\).  We have that
    \[\mathbf{P}_{y \sim \mathrm{Unif}(Y)} \left[ |\langle y - \mu_Y, v_\mu \rangle | \ge t \sigma \right] \le \frac{1}{t^2}.\]
    This means that there exists a subset \(\widehat{Y} \subseteq Y\) such that \(|\widehat{Y}| \ge \left(1 - \frac{1}{t^2}\right) |Y|\), and for all \(y \in \widehat{Y}\) we have \(|\langle y - \mu_Y, v_\mu \rangle | \le t \sigma\).  For each \(y \in \widehat{Y}\), we can decompose \(y - \mu_Y\) into the portion that is parallel to \(v_\mu\): \(\overline{(y - \mu_Y)} = \langle y - \mu_Y, v_\mu \rangle v_\mu\), and the portion that is perpendicular to \(v_\mu\): \((y - \mu_Y)^\perp = (y - \mu_Y) - \overline{(y - \mu_Y)}\).  Let \(\overline{(y - c)}\) be the portion of \(y - c\) that is parallel to \(v_\mu\): \(\overline{(y - c)} = \langle y - c, v_\mu \rangle v_\mu\).  By the choice of \(v_\mu\), we have that \(\overline{(y - c)} = \overline{(y - \mu_Y)}\).  
    Thus, 
    \[||y - \mu_Y||^2_2 = || (y - \mu_Y)^\perp ||_2^2 + ||\overline{(y - \mu_Y)}||^2_2 = || (y - c)^\perp ||_2^2 + ||\overline{(y - \mu_Y)}||^2_2 \le ||y - c||_2^2 + (t\sigma)^2.\]

    Instead of bounding the number of points that deviate significantly from \(\mu_Y\), we would like to bound the number of points that deviate significantly from \(\widehat{\mu}\).  We observe that a similar argument goes through.  We have that
    \begin{align}
        \mathbf{E}_{y \sim \mathrm{Unif}(Y)} \left[ (y - \widehat{\mu})(y - \widehat{\mu})^\top \right] &= \mathbf{E}_{y \sim \mathrm{Unif}(Y)} \left[ (y - \mu_Y + \mu_Y - \widehat{\mu})(y - \mu_Y + \mu_Y - \widehat{\mu})^\top \right] \nonumber \\
        &= \mathbf{E}_{y \sim \mathrm{Unif}(Y)} \left[ (y - \mu_Y)(y - \mu_Y)^\top \right] + \mathbf{E}_{y \sim \mathrm{Unif}(Y)} \left[ (y - \mu_Y) \right] (\mu_Y - \widehat{\mu})^\top \nonumber \\
        &\qquad \qquad + (\mu_Y - \widehat{\mu}) \mathbf{E}_{y \sim \mathrm{Unif}(Y)} \left[ (y - \mu_Y)^\top \right] + (\mu_Y - \widehat{\mu})(\mu_Y - \widehat{\mu})^\top \nonumber \\
        &= \Sigma_Y + (\mu_Y - \widehat{\mu})(\mu_Y - \widehat{\mu})^\top \nonumber \\
        &\preccurlyeq \Sigma_Y + \tau^2 I. \label{eq:variance-bound-chat}
    \end{align}
    Let \(v_{\widehat{\mu}} = \frac{\widehat{\mu}- c}{||\widehat{\mu} - c||}\) be a unit vector in the direction of \(\widehat{\mu} - c\).  By (\ref{eq:variance-bound-chat}) we have that 
    \[\mathbf{P}_{y \sim \mathrm{Unif}(Y)} \left[ |\langle y - \widehat{\mu}, v_{\widehat{\mu}} \rangle | \ge t \sqrt{\sigma^2 + \tau^2} \right] \le \frac{1}{t^2}.\]
    Thus there exists a subset \(\widehat{Y} \subseteq Y\) such that \(|\widehat{Y}| \ge \left(1 - \frac{1}{t^2}\right) |Y|\), and for all \(y \in \widehat{Y}\) we have \(|\langle y - \widehat{\mu}, v_{\widehat{\mu}} \rangle | \le t \sqrt{\sigma^2 + \tau^2}\).  For each \(y \in \widehat{Y}\), we can decompose \(y - \widehat{\mu}\) into the portion that is parallel to \(v_{\widehat{\mu}}\): \(\overline{(y - \widehat{\mu})} = \langle y - \widehat{\mu}, v_{\widehat{\mu}} \rangle v_{\widehat{\mu}}\), and the portion that is perpendicular to \(v_{\widehat{\mu}}\): \((y - \widehat{\mu})^\perp = (y - \widehat{\mu}) - \overline{(y - \widehat{\mu})}\).  Let \(\overline{(y - c)}\) be the portion of \(y - c\) that is parallel to \(v_{\widehat{\mu}}\): \(\overline{(y - c)} = \langle y - c, v_{\widehat{\mu}} \rangle v_{\widehat{\mu}}\).  By the choice of \(v_{\widehat{\mu}}\), we have that \(\overline{(y - c)} = \overline{(y - \widehat{\mu})}\).  Thus we have 
    \[||y - \widehat{\mu}||^2_2 = || (y - \widehat{\mu})^\perp ||_2^2 + ||\overline{(y - \widehat{\mu})}||^2_2 = || (y - c)^\perp ||_2^2 + ||\overline{(y - \widehat{\mu})}||^2_2 \le ||y - c||_2^2 + t^2(\sigma^2 + \tau^2).\]
\end{proof}

\subsection{Reducing Worst-Case Improper Learning to Non-Worst-Case Proper Learning}

Our goal is to use \Cref{lem:ball-bounded-variance-pts-near-mean} to get a guarantee for learning a small ellipsoid, that does not depend on the variance of the points.  As a stepping stone to our population statement, we first prove the following statement about the sample points.

\begin{restatable}[Finding small volume ellipsoid for $n$ points]{theorem}{smallevolellipsoidnpoints}
    Suppose we are given a set of points \(Y \subseteq \mathbb{R}^d, |Y| = n\), and \(0 \le \delta \le 1, 0 \le \gamma \le 1\), such that there exists a subset \(Y^\star \subseteq Y\), \(|Y^\star| \ge \delta |Y|\), that is contained in an unknown ball \(B^\star = B(c^\star, R^\star)\). 
    Then we can find an ellipsoid \(\widehat{E}\) such that \(|\widehat{E} \cap Y| \ge \delta(1-\gamma) |Y|\), and 
    $$\mathrm{vol}(\widehat{E})^{1/d} \le \mathrm{vol}(B^\star)^{1/d} \cdot \left(1 + O\left(d^{-1/2 + o(1)}/ \gamma \delta \right) \right).$$
    \label{thm:small-volume-ellipsoid}
\end{restatable}

Intuitively, the case where our non-worst-case proper learning algorithm fails is when the points have most of their variance concentrated in only a few directions, see for example \Cref{fig:ellipsoid}.  This makes it hard to estimate the position of the center accurately in the large variance directions.  However, this also means that the variance of the points in the other directions must be low, so it is \emph{easier} to estimate the position of the center in these directions.  Applying the correct linear transformation to the points essentially allows us to ``reweigh" the accuracy that we are aiming for across the different directions, and trade off error in the high variance directions for accuracy in the low variance directions.

\begin{figure}[h]
         \centering
    \includegraphics[width=0.5\linewidth]{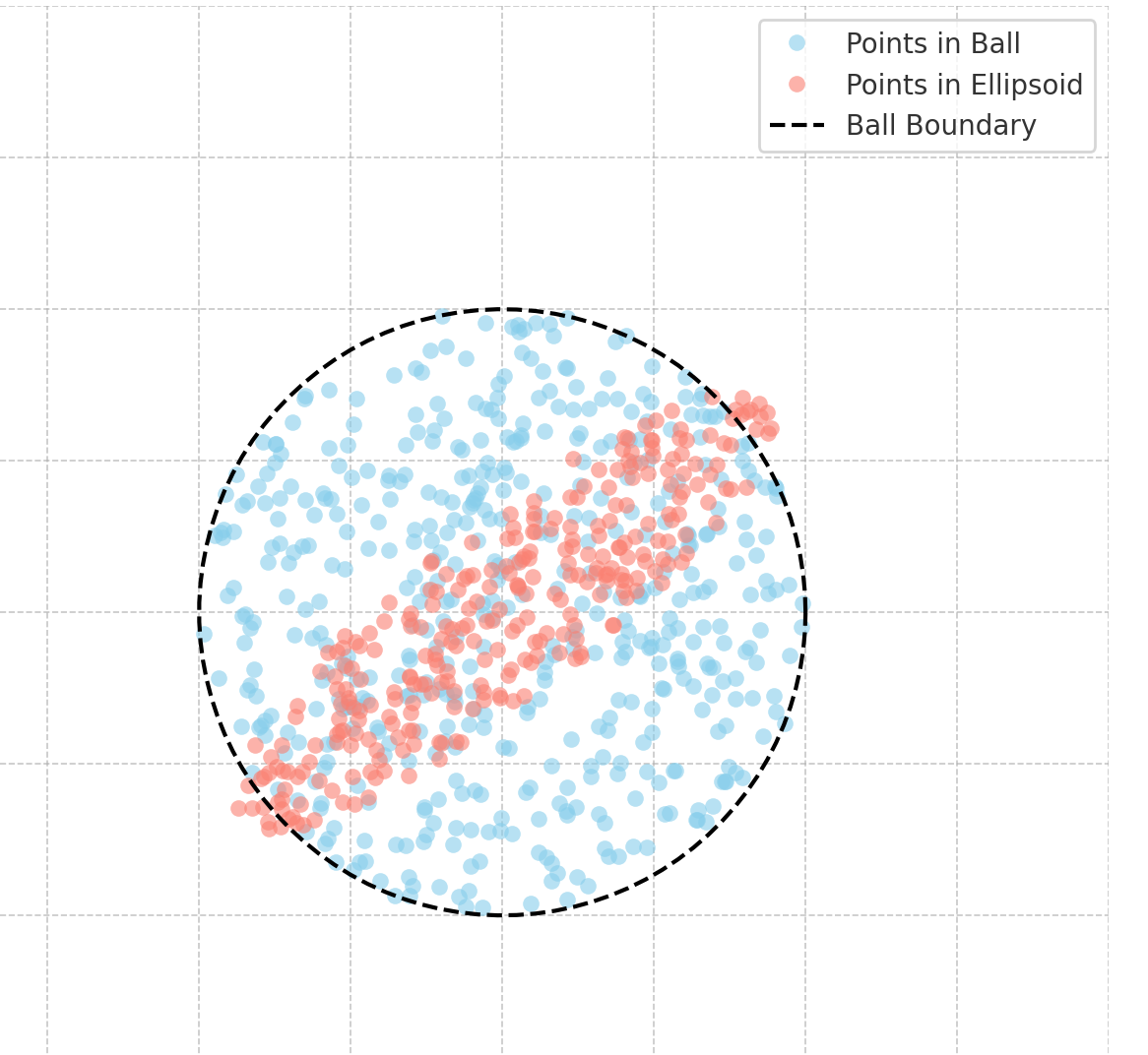}
    \caption{The figure shows a ball $B^{\star}$ of radius $R^{\star}=\sqrt{d}$ containing a set of points, whose the covariance matrix $\Sigma_{Y^{\star}}$ is not isotropic and has some directions of high variance. In this case, we can try to find a smaller ellipsoid containing many of the points. }
    \label{fig:ellipsoid}
\end{figure}

Our strategy is to find a good preconditioner \(\whM^{-1/2}\), for the points in \(Y\), such that the transformed points have small variance.  (Writing the preconditioner as \(\whM^{-1/2}\) is convenient, as we can then argue about the distortion in terms of the eigenvalues of \(\whM\), which we will choose to be positive definite.)  Then, we can apply our procedure to find an approximately minimum volume ball on the transformed points, which will return a strong approximation since the variance of the points is controlled (via \Cref{lem:ball-bounded-variance-pts-near-mean}).  Transforming this ball back into the original space will result in a confidence region that is shaped as an ellipsoid.  

Our goal is to choose \(\whM\) so that the largest variance of the transformed points in any one direction is comparable to the \emph{average} variance across all the directions (i.e., the variance is well-spread, and the distribution is ``approximately isotropic").  We choose \(\whM\) to have eigenvectors aligned with those of the covariance \(\Sigma_{Y'}\) of the points \(Y'\), where \(Y'\) is some subset of $Y$ contained in a ball of radius \(R\), so that the average variance across the \(d\) directions is at most \(R^2 / d\).  We choose \(\whM\) to have the same eigenvectors as \(\Sigma_{Y'}\), with eigenvalues $a_i^2$ of $\whM$ satisfying $a_i^2 = d$ if $\lambda_i \geq \tauhat^2 R^2/d$ and $a_i^2 = 1$ if $\lambda_i < \tauhat^2 R^2/d$, where \(\tauhat\) is a parameter that we choose later.  

Essentially, this means that \(\whM^{-1/2}\) will shrink the large variance directions by a factor \(\sqrt{d}\), ensuring that the resulting variance in those directions is constant (depending only on \(R\)).  It does nothing for the small variance directions.  This choice of \(\whM\) has the following nice properties.

\begin{enumerate}[(i)]
    \item \(\whM \succcurlyeq I\), and therefore \(\whM^{-1/2}\) is ``shrinking" (non-expanding) in all directions.  This ensures that there exists a ball in the transformed space that contains $Y^\star$ and has volume at most that of $B^\star$ in the original space. 

    \item We have that the variance of the points \(Y\) in the transformed space is bounded, i.e., the points in the transformed space are approximately isotropic:
    \[\whM^{-1/2}\Sigma_{Y}\whM^{-1/2} \preccurlyeq O\left(\frac{\tauhat^2 R^2}{d}\right)I_d.\]
    This allows us to bound the approximation ratio of applying \Cref{lem:ball-bounded-variance-pts-near-mean} to the transformed points.

    \item At most \(d/\tauhat^2\) eigenvalues of \(\whM\) are set to \(d\).  This bounds the distortion of \(\whM\).  That is, for any shape \(S\), 
    \[\frac{\mathrm{vol}(S)}{\mathrm{vol}(\whM^{-1/2}S)} \le d^{d/\tauhat^2} \le \exp \left( \frac{d \ln d}{\tauhat^2} \right).\]
\end{enumerate}
Properties (i) and (ii) allow us to reduce to the non-worst-case setting, and property (iii) ensures that this translation does not lose too much in volume. We will set $\tauhat= d^{1/4}$ to trade-off the different losses from (ii) and (iii) and obtain the final guarantee.  
The full proof of this theorem can be found in \Cref{sec:ellipsoid}.

\anote{Added 5/11. Could move it elsewhere e.g., to end of Section 1.2.} We remark that our main result (Theorem~\ref{thm:intro:ellipsoid}) gives a volume approximation guarantee of $\exp(d^{1/2+o(1)})$ in the worst-case. At the same time, Corollary~\ref{cor:ball-bounded-variance-better-bounds} gives a $O(1)$ factor approximation guarantee in volume, when the distribution of points inside the optimal ball $B^\star$ is approximately isotropic. We leave as an open question whether or not there exists a polynomial time (improper) learning algorithm that obtains an $O(1)$ volume approximation in the worst-case.

\section{Proper Learning of Euclidean Balls} \label{sec:balls}

The main technical insight that allows us to get a proper learning algorithm for Euclidean balls in a non-worst-case setting is \Cref{lem:ball-bounded-variance-pts-near-mean}, which tells us that for a set of points with bounded variance, the mean of the points is a good proxy for the center (up to a slack factor in coverage). 

We use this result as a subroutine in our main result for improper learning.  In this section we provide some additional results that extend our result for proper learning to other settings.  In particular, \Cref{thm:approx-optimal-ball} gives a proper learning algorithm for \emph{fully worst-case} instances, and achieves an approximation ratio of \(\left(1 + \widetilde{O}_{\gamma, \delta}\left( \frac{1}{\log d} \right) \right)\) in radius, which is comparable to the factor of \(\left(1 + \widetilde{O}_{\gamma, \delta}\left( \frac{1}{\log d} \right) \right)\) due to \citep{badoiu2002approximate, ding:LIPIcs.ESA.2020.38}.  While we do achieve a slightly weaker approximation guarantee in this setting, our techniques are significantly different than the coreset based techniques of the prior work, and extend well to non-worst-case settings, as evidenced by our improper learning result.  We also provide \Cref{cor:ball-bounded-variance-better-bounds}, which shows a way to achieve a better approximation ratio, assuming non-worst-case properties only on the \emph{inlier} points.  

\begin{theorem}[Finding approximate minimum volume $\delta$-coverage ball] \label{thm:main:ball}
    Given a set of points \(Y \subseteq \mathbb{R}^d \), \(|Y| = n\), and \(0 \le \delta \le 1\), \(0 \le \gamma \le 1\), such that there exists a subset \(Y^\star \subseteq Y\), \(|Y^\star| \ge \delta |Y|\), that is contained in an unknown ball \(B^\star = B(c^\star, R^\star)\), we can find a ball \(\Bhat = B(\chat, \Rhat)\) such that \(|Y \cap \Bhat| \ge (1 - \gamma) |Y^\star|\), and 
    \[\vol(\Bhat)^{1/d} \le \vol(B^{\star})^{1/d} \left(1 + O\left( \frac{\log \log d}{\gamma \delta \log d } \right) \right) \]
    in polynomial time in $n,d$. 
    \label{thm:approx-optimal-ball}
\end{theorem}

The algorithm for finding this small 
volume coverage ball is given in \Cref{fig:algo_ball}.
By using the fact that the VC-dimension of Euclidean balls in $d$-dimensions is at most $d+1$, we immediately get the following corollary. 

\begin{corollary}[Population version of \Cref{thm:main:ball}]\label{corr:ball:stat}
    Suppose $Y\subset\mathbb{R}^d$ is a set of $n$ i.i.d. samples generated from an arbitrary distribution \(\mathcal{D}\). For a target coverage \(\delta \in (0, 1)\) and a coverage approximation factor $\gamma\in(0,1)$ such that $n=\Omega(d)$,  in \(\mathrm{poly}(n, d)\) time, with high probability, we can find a ball \(\Bhat\) such that with high probability,
    \[\mathbf{P}_{y \sim \mathcal{D}} \left[ y \in \Bhat \mid Y\right] \ge \delta ,\]
    and 
        \[\vol(\Bhat)^{1/d} \le \vol(B^{\star})^{1/d} \left(1 + O\left( \frac{\log \log d}{\gamma \delta \log d} \right)  \right)\]
    where \(B^\star\) is the minimum volume ball that achieves at least \(\delta + \gamma + O(\sqrt{d/n})\) coverage over \(\mathcal{D}\). 
\end{corollary}

Our approach starts with the following structural observation about the true high-density ball \(B^\star\).  It is not in general true that the mean \(\mu^\star\) of the points \(Y^\star = B^\star \cap Y\) should be near the center \(c^\star\) of \(B^\star\).  However, it turns out that if the mean is far from the center, taking a ball centered at the mean of \(B^\star\) must still capture \emph{most} of the points in \(Y^\star\).  This is a simple consequence of Chebyshev's inequality--- if we are able to control the variance of the points in \(Y^\star\) (\Cref{lem:ball-bounded-variance-pts-near-mean}, \Cref{fig:ball:overview}).
Thus our problem reduces to the following.  Given only the full set of points \(Y\), we must estimate the mean of \(Y^\star\), while also controlling the variance of the points that we are taking the mean of.

The following simple lemma from robust statistics shows that the empirical mean of a set of points with bounded covariance that contains both inliers and outliers is close to the true mean.  
\begin{lemma}[High probability event has close-by mean]\citep[see e.g., Lemma 2.1 of][]{steinhardt2019lecture}
Let $X \in \R^d$ be a random variable with mean $\mu \in \R^d$ and covariance matrix $\Sigma$, where the largest eigenvalue satisfies $\lambda_{\max}(\Sigma) \leq \sigma^2$. Then for any event $E$ with probability at least $\delta>0$, the conditional expectation satisfies:
$$\| \E[X|E] - \mu \|_2 \le \sigma \cdot \sqrt{2 (1-\delta)/\delta}.$$
\label{lem:ball-high-prob-event-has-mean-closeby}
\end{lemma}


The following simple claim shows that the variance does not increase much when we condition on an event with non-negligible probability. 
\begin{claim}[High probability event has similar variance] 
    
Let \(X \in \R^d\) be a random variable whose maximum variance in any direction at most \(\sigma^2\), and let \(E\) be an event that occurs with probability at least \(\delta\).  Then the maximum variance in any direction of \(X\) conditioned on \(E\) is at most \(\frac{\sigma^2}{\delta}\).
\label{lem:ball-large-subset-has-similar-covariance}
\end{claim}

\begin{proof}
    This follows from a simple averaging argument.  Let \(\mu\) be the mean of \(X\).  We have that the variance of \(X\) is at most \(\sigma^2\).  That is, for all unit vectors \(v\), 
    \[\mathbf{E}\left[ \langle X - \mu, v \rangle^2 \right] \le \sigma^2.\]

    Now, we can decompose this into the variance conditioned on \(E\),
    \begin{align*}
        \sigma^2 &\ge \mathbf{E}\left[ \langle X - \mu, v \rangle^2 \right] \\
        &\ge \mathbf{E} \left[ \langle X - \mu, v \rangle^2 | E\right] \mathbf{P}[E] + \mathbf{E} \left[ \langle X - \mu, v \rangle^2 | \neg E\right] \mathbf{P}[\neg E] \ge \mathbf{E} \left[ \langle X - \mu, v \rangle^2 | E\right] \delta \\
        \frac{\sigma^2}{\delta} &\ge \mathbf{E} \left[ \langle X - \mu_E, v \rangle^2 | E\right],
    \end{align*}
    where in the last line \(\mu_E\) is the mean of \(X\) conditioned on \(E\), and the last line follows because the mean minimizes the squared deviation.  Thus we can conclude that the variance of \(X\) given \(E\) is bounded by \(\frac{\sigma^2}{\delta}\) in every direction.  
\end{proof}

To use the above observation to find an approximate center, we need to do two things.  First, we need to control the standard deviation of the points that we are considering.  In particular, to achieve a subconstant approximation of radius, we need that the standard deviation of our points is \(\frac{1}{\omega(1)} R^\star\).  Second, we need to be able to find \(\mu^\star\) approximately, which is akin to robust mean estimation.

In the first step,  we use a simple search procedure that gives us a coarse estimate of \(B^\star\).  That is, we find a ball \(B'\) such that \(Y^\star \subseteq B'\) and the radius \(R'\) of \(B'\) is bounded by \(R^\star \le R' \le 2 R^\star\).
In the worst case, we cannot expect that the variance is less than \((R')^2\), for example if all of the variance of \(Y'\) is concentrated in one direction. However, there can be at most \(\log d\) directions with variance higher than \((R')^2 / \log d\).  Thus we can find a list of candidate centers for \(Y^\star\) by estimating the location of the center separately in the high variance space and the low variance space.  In the high variance space, we can grid search to find a list of potential candidates for the center.   This is tractable since the high variance space has dimension \(\log d\).  In the low variance space we can estimate the location of \(\mu^\star\).

\begin{lemma}[Grid search]
    Suppose we are given a ball \(B' = B(0, R')\) containing a set of points \(Y' \in \mathbb{R}^{q}\), such that there exists a subset of points \(Y^\star \in Y'\) that are contained in a smaller ball \(B^\star = B(c^\star, R^\star)\) with $c^* \in B'$, where \(|Y^\star| \ge \delta |Y'|\).

    Then for a tolerance \(0<\tau < R'\), there exists a list \(L\) of size \(\left(\frac{2 + \tau/R'}{\tau/R'} \right)^q\), such that \(L\) contains a point \(\chat \in \R^q\) such that for every \(y \in Y^\star\),
    \[\|y - \chat\| \le \|y - c^\star\| + \tau.\]
    \label{lem:ball-grid-search}
\end{lemma}

\begin{proof}
    We first construct a \(\frac{\tau}{R'}\)-net for the unit ball in $\R^q$, denoted as \(L_0\), with size \(\left(\frac{2 + \tau/R'}{\tau/R'} \right)^q\). This ensures that every point in the unit ball is at most distance \(\frac{\tau}{R'}\) from some point in the net $L_0$.  Next, we scale the net $L_0$ to \(L = R' L_0\). This ensures that every point in \(B' = B(0, R')\) is at most distance \(\tau\) from some point in \(L\). Thus, for the center $c^* \in B'$, there is a point $\chat \in L$ such that $\|\chat - c^*\| \leq \tau$. By the triangle inequality, for every $y \in Y^*$,  we have $\|y - \chat\| \leq \|y - c^*\| + \|\chat - c^*\|  \leq \|y - c^*\| + \tau$.
\end{proof}

\begin{lemma}[Coarse-to-fine refinement]
    Suppose we are given a ball \(B' = B(0, R')\) containing a set of points \(Y' \subseteq \mathbb{R}^d\), where the \((q + 1)\)-th largest eigenvalue of the covariance matrix $\Sigma_{Y'}$ of \(Y'\), denoted as \(\lambda_{q+1}(\Sigma_{Y'})\), is bounded by 
    \[\lambda_{(q + 1)}(\Sigma_{Y'}) \le \sigma^2.\]
    Suppose there exists a subset of points \(Y^\star \subseteq Y'\) are contained in a smaller ball \(B^\star = B(c^\star, R^\star)\), where \(|Y^\star| \ge \delta |Y'|\) and \(R' \le 2 R^\star\).  Let \(\gamma > 0\) be a coverage slack factor. 

    Then there exists a list \(C\) of size at most \(\left( \frac{3 R'}{\tau} \right)^q\) and a set \(\widehat{Y} \subseteq Y^\star\) with \(|\Yhat| \ge (1 - \gamma) |Y^\star|\), such that there exists a point \(\chat \in C\) which for every \(y \in \Yhat\) satisfies 
    \begin{align*}
        ||y - \chat||_2^2 &\le ||\Pi_\low (y - c^\star)||_2^2 + \frac{\sigma^2}{\gamma}  \left(1 +  \frac{2 (1 - \delta) }{ \delta}\right) + \left(\|\Pi_\high (y - c^\star) \|_2 + \tau\right)^2 \\
        &\le \left( ||y - c^\star||_2 + O\left(\frac{\sigma}{\sqrt{\gamma \delta}} \right) + \tau \right)^2,
    \end{align*}
    where \(\Pi_\high\) is the projection onto the top-\(q\) eigenspace of \(\Sigma_{Y'}\) and \(\Pi_\low\) is the projection onto the space orthogonal to the top-\(q\) eigenspace of \(\Sigma_{Y'}\).  For \(q = 0\), we get a single estimate \(\chat\), and \(\widehat{Y} \subseteq Y^\star\) with \(|\Yhat| \ge (1 - \gamma) |Y^\star|\), that for every \(y \in \Yhat\) satisfies  
    \begin{align*}
        ||y - \chat||_2^2 &\le ||y - c^\star||_2^2 + \frac{\sigma^2}{\gamma}  \left(1 +  \frac{2 (1 - \delta) }{ \delta}\right).
    \end{align*}
    

    \label{lem:ball-coarse-to-fine}
\end{lemma}

\begin{proof}
    We start by considering \(\Sigma_{Y'}\), the covariance matrix of \(Y'\) and partition its eigenvalues into two groups.  
    The directions associated with the largest \(q\) eigenvalues of $\Sigma_{Y'}$ are referred to as the ``high variance directions"; while the directions corresponding to the remaining eigenvalues are the ``low variance directions".  Let \(\Pi_\high : \mathbb{R}^d \rightarrow \mathbb{R}^q\) denote the projection onto the ``high variance" eigenspace and \(\Pi_\low : \mathbb{R}^d \rightarrow \mathbb{R}^{d - q}\) denote the projection onto the ``low variance" eigenspace.
    We estimate the mean \(\mu^\star\) of \(Y^\star\) separately in the high variance and low variance directions.

    First, if \(q > 0\), we estimate the location of \(\mu^\star\) in the high variance directions.  
    Let \(Y'_\high = \Pi_\high Y'\) denote the projection of points in \(Y'\) onto the high variance directions.  
    Note that the set \(Y'_\high\) lies within the \(q\)-dimensional ball \(B'_\high = \Pi_\high B'\), which is centered at the origin and has a radius of at most \(R'\). 
    Next, consider the projection of $Y^*$ onto the high variance directions, denoted by \(Y^*_\high = \Pi_\high Y^*\). 
    Then, this set $Y^*_\high$ is contained in the ball \(B^*_\high = \Pi_\high B^* \) centered at \(\Pi_\high c^*\) with radius at most \(R^\star\).  Since \(B^\star_\high \subseteq B'_\high\), and \(\mu^\star \in B'_\high\) as well, 
    by \Cref{lem:ball-grid-search}, there exists a list \(L \subset \R^q\) with size \(|L| \le \left(\frac{2 + \tau/R'}{\tau/R'} \right)^q \le \left( \frac{3 R'}{\tau} \right)^q \) 
    such that there is a point \(\ellhat \in L\) such that 
    for every point \(y \in Y^\star\),
    \begin{equation} 
    \|\Pi_\high y - \ellhat\| \le \|\Pi_\high y - \Pi_\high c^*\| + \tau.
    \end{equation}

    Next, we estimate the location of \(\mu^\star\) in the low variance directions.  Let \(Y'_\low = \Pi_\low Y'\) denote the projection of $Y'$ onto the low variance directions.  
    By construction, the set \(Y'_\low\) has variance at most \(\lambda_{(q + 1)}(\Sigma_{Y'}) \le \sigma^2\) in every direction.
    Since the set $Y^*$ contains at least a $\delta$-fraction of the points in \(Y'\), by \Cref{lem:ball-large-subset-has-similar-covariance}, we have \(Y^\star_\low = \Pi_\low Y^\star\) has variance at most 
    $\sigma^2/\delta$ in every direction.
    By \Cref{lem:ball-high-prob-event-has-mean-closeby}, the mean \(\mu^\star_\low\) of \(Y^\star_\low\) is close to the mean \(\mu'_\low\) of \(Y'_\low\), i.e.,
    \begin{equation}
        ||\mu^\star_\low - \mu'_\low||_2 \le \sigma \cdot \sqrt{2 (1 - \delta) / \delta}.
        \label{eq:mustar-close-to-muprime}
    \end{equation}
    Since the variance of \(Y^\star_\low\) is bounded by \(\sigma^2/\delta\), and \(\mu'_\low\) is close to \(\mu^\star_\low\), \Cref{lem:ball-bounded-variance-pts-near-mean} tells us that, for a fixed choice of \(\gamma\), there exists a subset \(\widehat{Y} \subseteq Y^\star\) such that \(|\widehat{Y}| \ge (1 - \gamma) |Y^\star|\), and for every point \(y \in \widehat{Y}\), 
    \[||\Pi_\low y - \mu'_\low ||^2_2 
    \le \|\Pi_\low (y - c^\star)\|_2^2 + \frac{\sigma^2}{\gamma}  \left(1 +  \frac{2 (1 - \delta) }{ \delta}\right) . \]
    For \(q = 0\), this gives the final bound.

    For \(q > 0\), we take \(\chat\) such that \(\Pi_\high \chat = \ellhat\) and \(\Pi_\low \chat = \mu'_\low\).
    Putting the bounds together, we have that there exists a set \(\widehat{Y} \subseteq Y^\star\), with \(|\widehat{Y}| \ge (1 - \gamma) |Y^\star|\), such that for every \(y \in \widehat{Y}\),
    \begin{align*}
        \|y - \chat||_2^2 &= ||\Pi_\low(y - \chat)||_2^2 + ||\Pi_\high (y - \chat)||_2^2 \\
        &\le ||\Pi_\low (y - c^\star)||_2^2 + \frac{\sigma^2}{\gamma}  \left(1 +  \frac{2 (1 - \delta) }{ \delta}\right) + \left(\|\Pi_\high (y - c^\star) \|_2 + \tau\right)^2 \\
        &\le \left( ||y - c^\star||_2 + O\left(\frac{\sigma}{\sqrt{\gamma \delta}} \right) + \tau \right)^2,
    \end{align*}
    where both the tighter bound (second line) and looser bound (third line) will be useful in applications of the lemma.

\end{proof}

\begin{lemma}[Most directions are low variance] 
    Let \(Y\) be a set of points contained in \(B(0, R)\) with mean \(\mu_Y\).  The covariance of \(Y\) can have at most \(q\) eigenvalues greater than \(\frac{R^2}{q}\).
    \label{lem:ball-most-directions-low-variance}
\end{lemma}

\begin{proof}
   Denote the covariance of \(Y\) by \(\Sigma_Y = \frac{1}{|Y|} \sum_{y \in Y} (y - \mu_Y) (y - \mu_Y)^\top\), where \(\mu_Y\) is the mean of the points \(Y\).  We have that 
   \[\Sigma_Y \preccurlyeq \frac{1}{|Y|} \sum_{y \in Y} y y^\top,\]
   since the mean minimizes the deviation.  Since trace follows the PSD domination, we have that 
   \[\mathrm{tr}(\Sigma_Y) \le \mathrm{tr}\left( \frac{1}{|Y|} \sum_{y \in Y} y y^\top \right).\]
   Since trace is a linear operator and $Y$ is contained in $B(0,R)$, we have that 
   \[\mathrm{tr}\left( \frac{1}{|Y|} \sum_{y \in Y} y y^\top \right) =  \frac{1}{|Y|} \sum_{y \in Y} \mathrm{tr}\left(y y^\top \right) = \frac{1}{|Y|} \sum_{y \in Y} \|y\|_2^2 \le R^2. \]
   So we have that \(\mathrm{tr}(\Sigma_Y) \le R^2\).  Since the trace is equal to the sum of the eigenvalues, and \(\Sigma_Y\) is positive semidefinite and thus has all non-negative eigenvalues, we have by an averaging argument that \(\Sigma_Y\) can have at most \(q\) eigenvalues \(\ge \frac{R^2}{q}\).
\end{proof}

\subsection{Completing the Proof of Theorem \ref{thm:main:ball}}

\begin{figure}[ht]
\begin{tcolorbox}
\begin{center}
    \textbf{Algorithm} $\textsc{Dense\_Ball}$
\end{center}
\textbf{Input:} a set of point $Y \in \R^d$, target fraction $\delta \in (0,1)$, and  slack $\gamma \in (0,1)$\\
\textbf{Output:} an ball $\Bhat \subset \R^d$
\begin{enumerate}
\item Create a list of coarse balls $\mathfrak{B} = \left\{ B(y_1, \|y_1 - y_2\|) ~|~ y_1, y_2 \in Y \right\}$ and remove from $\mathfrak{B}$ all balls that contain less than $\delta n$ points in $Y$.
\item Set $R_{\min}$ to be the minimum radius among all balls in $\mathfrak{B}$. Then remove from $\mathfrak{B}$ all balls with a radius greater than $2R_{\min}$.

\item Create a list of refined balls $\mathfrak{B}_2$ as follows: 
\begin{enumerate}
    \item Create a list of candidate centers $C_{B'}$ for each ball $B'$ in $\mathfrak{B}$ as follows.
    \item Compute the eigenvalues $\lambda_1, \dots, \lambda_d$ and eigenvectors $v_1,\dots, v_d$ of the covariance $\Sigma_{Y'}$ of points $Y' = Y \cap B'$.
    \item Set $\Pi_{\high}$ be the projection onto the eigenspace of $\Sigma_{Y'}$ with eigenvalues greater than $(R')^2 \log\log d/ \log d$ and $\Pi_\low = I - \Pi_\high$. 
    Set $B'_{\high} = \Pi_{\high} B'$ and  $Y'_{\low} = \Pi_{\low} Y'$. 
    Note that there are at most $q = \log d/ \log\log d$ eigenvalues of $\Sigma_{Y'}$ is greater than $(R')^2 \log\log d/ \log d$. When $q>0$, we run a grid search over these high-variance directions; however, when $q=0$, no grid search is needed.
    \item Set the list $C_{B'} = \{c : \Pi_{\high}c = \ell, \Pi_{\low} c = \mu'_\low, \ell \in L\}$, where $\mu'_\low$ is the mean of $Y'_\low$ and $L$ is the $\tau$-net of the ball $B'_\high$ for $\tau = R_{\min}/\log d$.
    \item Combine all candidate centers $\mathfrak{C} = \cup_{B' \in \mathfrak{B}} C_{B'}$.
    \item Set the list $\mathfrak{B}_2 = \left\{ B(c, ||y - c||) ~|~ c \in \mathfrak{C}, y \in Y \right\}$.
\end{enumerate}
\item Return the ball $\Bhat \in \mathfrak{B}_2$ with the smallest volume that contains at least $\delta(1-\gamma)$ points in $Y$.
\end{enumerate}
\end{tcolorbox}
\caption{Algorithm $\textsc{Dense\_Ball}$ for finding a small volume ball that contains at least $\delta' = \delta(1-\gamma)$ fraction of points}
\label{fig:algo_ball}
\end{figure}

\begin{proof}[Proof of \Cref{thm:main:ball}]
    We begin by finding a list of coarse guesses for \(B^\star\).  We create a list of balls 
    \[\mathfrak{B} = \left\{ B(y_1, \|y_1 - y_2\|) ~|~ y_1, y_2 \in Y \right\} .\]
    Note that \(|\mathfrak{B}| \le n^2\) since $|Y| = n$. Consider the subset $Y^* \subseteq Y$ contained in ball $B^* = B(c^*,R^*)$. Choosing 
    \[y_1^\star, y_2^\star = \argmax_{y_1, y_2 \in Y^\star} \|y_1 - y_2\|,\]
    gives us that the ball 
    \begin{equation} 
        B' = B(y_1^\star, \|y_1^\star - y_2^\star\|) 
        \label{eq:good-coarse-ball}
    \end{equation}
    contains all points in \(Y^\star\), has radius \(R' \le 2 R^\star\), and \(B' \in \mathfrak{B}\).  

    We now do a filtering step to get a coarse estimate of \(R^\star\), and remove any balls from \(\mathfrak{B}\) that are too large.  In particular, let \(R_\mathrm{min}\) be the minimum radius among balls in \(\mathfrak{B}\) that contain at least \(\delta n\) points. Since $B^*$ is the volume optimal ball that contains at least $\delta$-fraction of $Y$, we have $R_{\min} \geq R^*$. Thus, we know that 
    \[R^\star \le R_\mathrm{min} \le 2 R^\star.\] 
    We filter \(\mathfrak{B}\) to only contain balls that contain at least \(\delta n \) points from \(Y\), and have radius at most \(2 R_\mathrm{min} \le 4 R^\star\).  Finding \(R_\mathrm{min}\) and filtering the list take \(\widetilde{O}(n^3)\) time.

    Now we apply the coarse-to-fine refinement of \Cref{lem:ball-coarse-to-fine} to each ball in (the filtered) \(\mathfrak{B}\), to generate a list of candidate centers for \(\Bhat\).  In particular, we apply the procedure for \(q = \frac{\log d}{\log \log d}\) and \(\tau = \frac{R_\mathrm{min}}{\log d}\).  For each ball $B$ in the filtered $\mathfrak{B}$, we have that the radius \(R\) of \(B\) is at most \(2 R_\mathrm{min}\).  
    Thus, by \Cref{lem:ball-coarse-to-fine}, the size of the list of potential centers $C_B$ that we generate for \(B\) is at most
    \[|C_B| \le \left( \frac{3R'}{\tau} \right)^q \le \left( 6 \log d \right)^{\frac{\log d}{\log \log d}} \le d^{3} .\]
    Aggregating over the up to \(n^2\) balls in \(\mathfrak{B}\), we collect a list \(\mathfrak{C} = \bigcup_{B \in \mathfrak{B}} C_B\) of most \(n^2 d^{\frac{7}{2}}\) potential centers.  

    Now, we do a more refined search for \(B^\star\) using the candidate centers in \(\mathfrak{C}\).  That is, we generate the list 
    \[\mathfrak{B}_2 = \left\{ B(c, ||y - c||) ~|~ c \in \mathfrak{C}, y \in Y \right\}.\]
    Note that \(|\mathfrak{B}_2| \le |\mathfrak{C}| \cdot n \le n^3 d^{\frac{7}{2}}\).  Then, we output the minimum volume ball \(B \in \mathfrak{B}_2\), such that \(|B \cap Y| \ge \delta (1 - \gamma) n\).  Finding this ball can be done in time \(\widetilde{O}(|\mathfrak{B}_2| \cdot n) = \widetilde{O}(n^4 d^{\frac{7}{2}})\).  Thus the total runtime of our algorithm is bounded by \(\widetilde{O}(n^4 d^{\frac{7}{2}})\).
    
    Now we argue that \(\mathfrak{B}_2\) contains a ball that is a good approximation to \(B^\star\) in volume.  Consider the ball \(B'\) from (\ref{eq:good-coarse-ball}).  We have that \(B'\) contains all points in \(Y^\star\), has radius \(R' \le 2 R^\star\), and is in the list \(\mathfrak{B}\).  Let \(Y' = Y \cap B'\) be the subset of points of \(Y\) that are captured by \(B'\).  An averaging argument (\Cref{lem:ball-most-directions-low-variance}) tells us that \(\lambda_{(q + 1)}(\Sigma_{Y'}) \le \frac{R'^2}{q}\), where \(\Sigma_{Y'}\) is the covariance of the points \(Y'\), and \(\lambda_{(q + 1)}(\Sigma_{Y'})\) is the \((q + 1)\)-th largest eigenvector of \(\Sigma_{Y'}\).  Thus, when we apply the procedure of \Cref{lem:ball-coarse-to-fine} to \(B'\) with \(q = \frac{\log d}{\log \log d}\) and \(\tau = \frac{R_\mathrm{min}}{\log d}\), we also have that the eigenvalue $\lambda_{(q + 1)}(\Sigma_{Y'})$ is upper bounded by \(\sigma^2 \le \frac{R'^2}{q} \le  \frac{(2R^\star)^2 \log \log d}{\log d}\).  This guarantees that, for the list of potential centers \(C_{B'}\) that we get for this ball \(B'\), there is a center \(\chat \in C_{B'}\) and a subset \(\Yhat \subseteq Y^\star\), \(|\Yhat| \ge (1 - \gamma) |Y^\star|\), such that for every \(y \in \Yhat\), 
    \begin{align*}
        \| y - c^\star\|_2^2 &\le \|\Pi_\low (y - c^\star)\|_2^2 + O\Big(\frac{\sigma^2}{ \gamma \delta}\Big) + (\|\Pi_\high (y - c^\star) \|_2 + \tau)^2, 
    \end{align*}
    where \(\Pi_\high\) is the projection onto the top-\(q\) eigenspace of \(\Sigma_{Y'}\) and \(\Pi_\low\) is the projection onto the space orthogonal to the top-\(q\) eigenspace of \(\Sigma_{Y'}\).

    To bound the second term, we use the fact that 
    \[(a + b)^2 \le (1 + \varepsilon) a^2 + (1 + \frac{1}{\varepsilon}) b^2, ~\forall \varepsilon > 0.\]
    Using \(\tau = \frac{R_\mathrm{min}}{\log d}\), and \(\varepsilon = \frac{1}{\log d}\), we get 
    \begin{align*}
        \| y - c^\star\|_2^2 &\le \|\Pi_\low (y - c^\star)\|_2^2 + O\left( \frac{R'^2 \log \log d}{\gamma \delta \log d} \right) + \left(1 + \frac{1}{\log d} \right) \|\Pi_\high (y - c^\star) \|_2^2 + (1 + \log d) \cdot \frac{R_\mathrm{min}^2}{\log^{2}d} \\
        &\le \left(1 + \frac{1}{\log d} \right) \|y - c^\star \|_2^2 + O\left( \frac{(R^\star)^2 \log \log d}{\gamma \delta \log d} \right) \\
        &\le (R^\star)^2 \left( 1 + O\left( \frac{\log \log d}{\gamma \delta \log d} \right)\right).
    \end{align*}
    Thus, it is guaranteed that the list of potential centers \(C_{B'}\) that we get for this ball $B'$ contains a center \(\chat \in C_{B'}\) such that \(\Bhat = B(\chat, \Rhat)\) contains at least \((1 - \gamma)|Y^\star|\) points for
    \begin{align*}
        \Rhat &\le R^\star \left( 1 + O\left( \frac{\log \log d}{\gamma \delta \log d} \right)\right)^{1/2},
    \end{align*}
    which corresponds to a volume bound of 
    \begin{align*}
        \vol(\Bhat) &\le \vol (B^\star) \exp \left( O \left(\frac{d \log \log d}{\gamma \delta \log d} \right) \right) .
    \end{align*}

    Since \(\chat \in \mathfrak{C}\), we are guaranteed that the smallest ball \(B_2 = B(\chat, ||y - c|| )\) for \(y \in Y\) that contains at least \((1 - \gamma)|Y^\star|\) points can have radius at most that of \(\Bhat\).  Since \(B_2 \in \mathfrak{B}_2\), we know that the minimum volume ball in \(\mathfrak{B}_2\) that contains at least \(\delta(1 - \gamma)n\) points from \(Y\) can only have radius at most that of \(B_2\), which is at most \(\Rhat\).


    
\end{proof}

\subsection{Better volume approximation guarantee for almost isotropic distributions}

Our strategy starts by finding a ball in \(B'\) that contains all of the points in \(B^\star\) and has radius at most twice that of \(B^\star\).  (More precisely, we find a list of candidate balls, which can be done simply.)  Trivially, this means that in every direction the standard deviation of the points in \(B'\) is at most \(R^\star\).  \Cref{lem:ball-coarse-to-fine} tells us that choosing the mean of the points in \(B'\) as the center of \(\Bhat\) essentially allows us to capture most of the points in \(B^\star\) with radius that is \(R^\star\) plus the maximum standard deviation in any direction.  This gives a \(2\)-approximation to the radius.  We note that even a slightly stronger bound on the standard deviation can improve the approximation factor to \(1 + o(1)\).   In \Cref{thm:approx-optimal-ball}, we do this by grid-searching the \(\log d\) many highest variance directions of the points in \(B'\), and taking the mean in the low-variance directions.  This allows us to argue about the standard deviation of the \((\log d + 1)\)th highest variance direction, which can be at most \(\frac{R^\star}{\sqrt{\log d}}\) (\Cref{lem:ball-most-directions-low-variance}).

This bound on the variance that we use in \Cref{thm:approx-optimal-ball} is the best that our approach achieves, since it is possible that the variance of the points is concentrated in a few directions.  However, if the points in \(B^\star\) (the ``inliers") are not worst-case, and are instead approximately isotropic, we expect the standard deviation in any direction to be bounded by \(\frac{R^\star}{\sqrt{d}}\), which we can use to give a better bound on the radius of the ball that we output. In our earlier argument, we bounded the variance not only of the inliers-- the points in \(B^\star\)-- but also all of \(Y'\), which includes some of the outliers.  To utilize this weaker assumption on the variance, we appeal to algorithms for list-decodable robust mean estimation, which work exactly in this setting where the inliers are well-behaved, but the outliers may be arbitrary.

This observation is also useful as a subroutine for our algorithm that finds a small volume ellipsoid in \Cref{sec:ellipsoid}.
At a high level, that algorithm will first estimate a linear transformation that limits the variance of the points in \(B^\star\).  Then, after applying this transformation, we essentially use Algorithm~\textsc{Dense\_Ball\_Isotropic} to search for a ball in the transformed space with a better volume approximation.  In that application, we will apply the transformation to and control the variance of all of our points, inliers and outliers, and thus avoid the need for list-decodable robust mean estimation.

\begin{figure}[ht]
\begin{tcolorbox}
\begin{center}
    \textbf{Algorithm} $\textsc{Dense\_Ball\_Isotropic}$
\end{center}
\textbf{Input:} a set of point $Y \in \R^d$, target fraction $\delta \in (0,1)$, slack $\gamma \in (0,1)$, isotropic parameter $\beta \in (0,1)$\\
\textbf{Output:} an ball $\Bhat \subset \R^d$
\begin{enumerate}
\item Create a list of coarse balls $\mathfrak{B} = \left\{ B(y_1, \|y_1 - y_2\|) ~|~ y_1, y_2 \in Y \right\}$ and remove from $\mathfrak{B}$ all balls that contain less than $\delta n$ points in $Y$.
\item Set $R_{\min}$ to be the minimum radius among all balls in $\mathfrak{B}$. 

\item Create a list of candidate centers $L$ by running the list-decodable mean estimation algorithm~\citep{DKKLT21, ars_book} on points $Y$, target fraction $\delta$, and variance $\beta(R_{\min})^2$.

\item Create a list of refined balls $\mathfrak{B}_2 = \left\{ B(c, ||y - c||) ~|~ c \in L, y \in Y \right\}$.

\item Return the ball $\Bhat \in \mathfrak{B}_2$ with the smallest volume that contains at least $\delta(1-\gamma)$ points in $Y$.
\end{enumerate}
\end{tcolorbox}
\caption{Algorithm $\textsc{Dense\_Ball\_Isotropic}$ for finding a small volume ball that contains at least $\delta' = \delta(1-\gamma)$ fraction of points for isotropic distributions}
\label{fig:algo_ball_isotropic}
\end{figure}

\boundedvarianceimpliesbetterbounds*


\begin{proof}
    We begin by finding a coarse estimate of \(R^\star\).  We do this as in \Cref{thm:approx-optimal-ball}: we can create a list of coarse guesses for \(B^\star\):
    \[\mathfrak{B} = \{ B(y_1, ||y_1 - y_2||) ~|~ y_1, y_2 \in Y \}.\]
    By taking \(y_1, y_2\) to be maximally distant points in \(B^\star\), we have that there is a \(\Bhat \in \mathfrak{B}\) such that the radius \(\Rhat\) of \(\Bhat\) satisfies 
    \[R^\star \le \Rhat \le 2 R^\star.\]
    Taking the minimum radius \(R_\mathrm{min}\) over all balls in \(\mathfrak{B}\) that covers at least $\delta n$ points gives us an estimate of \(R^\star\) such that \(R^\star \le R_\mathrm{min} \le \Rhat\le 2 R^\star\).

\anote{Edited the line below:}    Let \(\mu^\star\) be the mean of the points in \(Y^\star\).  We know that $Y^\star$ has covariance p.s.d. dominated by \( \tfrac{\beta (R^\star)^2}{d} I\), which is in turn p.s.d. dominated by \(\tfrac{\beta (R_\mathrm{min})^2}{d} I\). The polynomial-time algorithm for list-decodable mean estimation \citep{DKKLT21, ars_book} outputs a list \(L\) of length \(O(\log(1 / p) / \delta) \) such that with probability \(1 - p\), there exists a \(\widehat{\mu} \in L\) such that 
    \[ ||\widehat{\mu} - \mu^\star || = O\Big(\Big(\frac{\beta}{\delta d}\Big)^{1/2} R_{\mathrm{min}}\Big) = O\Big(\Big(\frac{\beta}{\delta d}\Big)^{1/2} R^\star \Big).\]


    By \Cref{lem:ball-bounded-variance-pts-near-mean} we have that there exists a subset \(\Yhat \subseteq Y \cap B^\star\), such that \(|\Yhat| \ge (1 - \gamma) |Y \cap B^\star|\), and for points \(y \in \Yhat\),
    \vnote{please double check that the big-O addition here is kosher.}\anote{Looks fine to me. Edited just some math typography.}
    \begin{align*}
        \|y - \widehat{\mu} \|^2_2 &\le \|y - c^\star\|_2^2 + \gamma^{-1} \left(\frac{\beta}{d} (R_{\mathrm{min}})^2 +  O \left(\delta^{-1} (\beta/d) (R^\star)^2 \right) \right) \\
        &\le (R^\star)^2 + O \left( \gamma^{-1}\delta^{-1} (\beta/d)(R^\star)^2 \right) \\
        &\le (R^\star)^2 \left(1 + O \left(  \frac{\beta}{\gamma \delta d}\right)\right) 
    \end{align*}

    Now, we can do one more search step and create a refined list of candidate balls 
    \[\mathfrak{B}_2 = \{B(\widehat{\mu}, ||\widehat{\mu} - y||) ~|~ \widehat{\mu} \in L, y \in Y\}.\]
    Since \(|\mathfrak{B}_2| \le |L| \cdot |Y|\), we have that this list is polynomially sized.
    We output the minimum volume ball \(\Bhat \in \mathfrak{B}_2\) such that \(|\Bhat \cap Y| \ge (1 - \gamma) \delta |Y|\).  This guarantees that \(\Bhat\) has radius \(\Rhat\) with 
    \[\Rhat \le R^\star \sqrt{1 +  O \left( \frac{\beta}{\gamma \delta d} \right)}. \]
\end{proof}



\section{\texorpdfstring{Finding Confidence Sets \((\Gamma = 1 + d^{-1/2 + o(1)})\)-Competitive with Euclidean Balls}{Finding Confidence Sets Competitive with Euclidean Balls}}\label{sec:ellipsoid}

In this section, we provide an algorithm to find a confidence set whose volume is within a\\ \(\exp (O(d^{1/2 + o(1)}))\)-factor of the volume of the smallest ball that contains a \(\delta\) fraction of \(Y\).

\smallvolumeellipsoid*

The above is an immediate consequence of the following theorem which is for the empirical version of the problem, using standard concentration tools \citep{devroye2001combinatorial} by incurring an additive $O(\sqrt{d^2/n})$ term in the coverage probability (since the VC-dimension of $d$-dimensional ellipsoids is at most $d^2+d$).

\smallevolellipsoidnpoints*

\begin{figure}[ht]
\begin{tcolorbox}
\begin{center}
    \textbf{Algorithm} $\textsc{Dense\_Ellipsoid}$
\end{center}
\textbf{Input:} a set of point $Y \in \R^d$, target fraction $\delta \in (0,1)$, and  slack $\gamma \in (0,1)$\\
\textbf{Output:} an ellipsoid $\whE \subset \R^d$
\begin{enumerate}
\item Create a list of coarse balls $\mathfrak{B} = \left\{ B(y_1, \|y_1 - y_2\|) ~|~ y_1, y_2 \in Y \right\}$ and remove from $\mathfrak{B}$ all balls that contains less than $\delta n$ points in $Y$.
\item Create an ellipsoid $\whE_i$ for each coarse ball $B'_i$ in the list $\mathfrak{B}$ as follows: 
\begin{enumerate}
    \item Set $\tauhat = d^{1/4}$ and $R'_i$ be the radius of $B'_i$.
    \item Compute the eigenvalues $\lambda_1,\dots, \lambda_d$ and eigenvectors $v_1,\dots, v_d$ of the covariance $\Sigma_{Y'_i}$ of points $Y'_i = Y \cap B'_i$.
    \item Create an ellipsoidal shape $\whM_i = \sum_{j=1}^d a_j^2 v_jv_j^\top$, where $a_j^2 = d$ if $\lambda_j \geq \tauhat^2 (R'_i)^2 /d$ and $a_i^2 = 1$ if $\lambda_j < \tauhat^2 (R'_i)^2 /d$.
    \item Compute the transformed points $Z_i = \whM_i^{-1/2} Y$. 
    Find a ball \(\Bhat_i\) using subroutine from Lemma \ref{lem:ball-coarse-to-fine} for \(q = 0\), \(\lambda_{(q + 1)}(\Sigma_{Y'}) \le \sigma^2\), \(\sigma^2 = \tauhat^2 (R')^2 / d\), slack factor \(\gamma\).
    \item Set $\whE_i = \whM_i^{1/2} \Bhat$.
\end{enumerate}
\item Return the ellipsoid $\whE$ with the smallest volume in the list $\{\whE_i\}$.
\end{enumerate}
\end{tcolorbox}
\caption{Algorithm $\textsc{Dense\_Ellipsoid}$ for finding a small volume ellipsoid that contains at least $\delta' = \delta(1-\gamma)$ fraction of points}
\label{fig:algo_ellipsoid}
\end{figure}

\begin{proof}[Proof of \Cref{thm:small-volume-ellipsoid}]
We begin by finding a list of coarse guesses for \(B^\star\), as in \Cref{thm:approx-optimal-ball}.  
We create a list of balls 
    \[\mathfrak{B} = \left\{ B(y_1, \|y_1 - y_2\|) ~|~ y_1, y_2 \in Y \right\} ,\]
    (step 1 of Algorithm~\textsc{Dense\_Ellipsoid}).
    Note that \(|\mathfrak{B}| \le n^2\) since $|Y| = n$. Consider the subset $Y^\star \subseteq Y$ contained in ball $B^\star = B(c^\star,R^\star)$. Choosing 
    \[y_1^\star, y_2^\star = \argmax_{y_1, y_2 \in Y^\star} \|y_1 - y_2\|,\]
    gives us that the ball 
    \begin{equation*} 
        B' = B(y_1^\star, \|y_1^\star - y_2^\star\|) 
    \end{equation*}
contains all points in \(Y^\star\). It has radius \(R' \le 2 R^\star\) and \(B' \in \mathfrak{B}\).  Let $Y'$ be all points contained in this ball $B'$.

For a parameter \(\tauhat\) that we will set later, Algorithm~\textsc{Dense\_Ellipsoid} chooses an ellipsoidal shape \(\whM\) to have the same eigenvectors as \(\Sigma_{Y'}\), with eigenvalues $a_i^2$ of $\whM$ satisfying $a_i^2 = d$ if $\lambda_i \geq \tauhat^2 (R')^2/d$ and $a_i^2 = 1$ if $\lambda_i < \tauhat^2 (R')^2/d$.  
\(\whM\) has the following nice properties.
\begin{enumerate}[(i)]
    \item The eigenvalues of \(\whM\) are in \(\{1, d\}\).  Thus \(\whM \succeq I\), and in particular, for any set of points \(P\), 
    \begin{equation}
        \mathrm{vol}(\mathrm{encball}(\whM^{-1/2} P)) \le \mathrm{vol}(\mathrm{encball}(P)),
        \label{eq:volume-shrinking}
    \end{equation}
    where \(\mathrm{encball}(S)\) is the minimum volume ball enclosing \(S\).
    
    \item Let $\Sigma_{Y'}$ be the covariance matrix of $Y'$, and $\lambda_1, \dots, \lambda_d$ be the eigenvalues of $\Sigma_{Y'}$.  
    We show that 
    \begin{equation}
    a_i^2\geq \frac{d}{\tauhat^2(R')^2}\lambda_i \qquad \text{ for all } i.
    \label{eq:psd-domination}
    \end{equation}
    When \(\lambda_i < \tauhat^2 (R')^2/d\), we set \(a_i^2 = 1\), which satisfies (\ref{eq:psd-domination}).  When \(\lambda_i \geq \tauhat^2 (R')^2/d\), we set \(a_i^2 = d\).  We know that since \(Y' \subseteq B'\), every eigenvalue of \(\Sigma_{Y'}\) is upper bounded as \(\lambda_i \le (R')^2\).  Thus \((d / \tauhat^2(R')^2 ) \lambda_i \le d / \tauhat^2 \le d\), satisfying (\ref{eq:psd-domination}).  By definition, (\ref{eq:psd-domination}) implies that 
    \[\whM \succeq \frac{d}{\tauhat^2 (R')^2} \Sigma_{Y'} .\]

    In particular, this implies that the points in \(Y'\) have bounded variance in the transformed space \(\whM^{-1/2}(\R^d) = \{\whM^{-1/2}y: y \in \R^d\} \), 
    \begin{equation}
        \whM^{-1/2}\Sigma_{Y'}\whM^{-1/2} \preccurlyeq O\left(\frac{\tauhat^2(R')^2}{d}\right)I_d.
        \label{eq:yprime-bounded-variance-in-transformed-space}
    \end{equation}

    \item Since \(Y' \subseteq B'\), the total variance of the points \(Y'\) is bounded as 
    \[\sum_i \lambda_i \le (R')^2.\]
    Since the \(\lambda_i\) are non-negative, this tell us that there can be at most \(d / \tauhat^2\) values of \(i\) for which \(\lambda_i \ge \tauhat^2 (R')^2 / d.\)  Thus at most \(d / \tauhat^2\) eigenvalues of \(\whM\) are set to \(d\).

    In particular, this implies that for any shape \(S\), 
    \begin{equation}
        \frac{\mathrm{vol}(S)}{\mathrm{vol}(\whM^{-1/2}S)} \le d^{d/\tauhat^2} \le \exp \left( \frac{d \ln d}{\tauhat^2} \right).
        \label{eq:ratio-transformed-volume}
    \end{equation}
\end{enumerate}


The algorithm transforms the points \(Y'\) by \(\whM^{-1/2}\).  We denote the points in the transformed space as \(T'\),
\[T' = \{ \whM^{-1/2} y ~|~ y \in Y' \}.\]
Recall that by the choice of \(Y'\), \(Y^\star \subseteq Y'\), and we denote 
\[T^\star = \{ \whM^{-1/2} y ~|~ y \in Y^\star \} \subseteq T'.\]

We apply the procedure from \Cref{lem:ball-coarse-to-fine} with \(q = 0\) to the minimum volume ball \(B_T' = B(c_T', R_T')\) containing \(T'\).  We have that \(\lambda_{\mathrm{max}}(\Sigma_{Y'}) \le \sigma^2\), where \(\sigma^2 = \tauhat^2(R')^2 / d\) (due to (\ref{eq:yprime-bounded-variance-in-transformed-space})), and slack factor \(\gamma\).  Let \(B_T^\star = B(c_T^\star, R_T^\star)\) be the minimum volume ball containing \(T^\star\).  \Cref{lem:ball-coarse-to-fine} allows us to find a ball \(\widehat{B}_T = B(\widehat{c}_T, \widehat{R}_T)\) that contains a subset \(\widehat{T} \subseteq T'\), where \(|\widehat{T}| \ge (1 - \gamma) |T^\star|\), and 
\begin{equation}
    \widehat{R}_T 
    \le \sqrt{(R_T^\star)^2 + \frac{\sigma^2}{\gamma} \left(1 + \frac{2(1 - \delta)}{\delta} \right)}.
    \label{eq:bthat-volume}
\end{equation}
(\Cref{lem:ball-coarse-to-fine} has an additional parameter also named \(\tau\), that is different from \(\tauhat\), but this parameter does not apply here since we set \(q = 0\) and avoid the grid search element of the procedure entirely.)  

Finally, we transform \(\widehat{B}_T\) back to the transformed space to get the candidate ellipsoid \(\widehat{E}\),
\[\widehat{E} = \whM^{1/2} \widehat{B}_T.\]
Note that \(| Y' \cap \widehat{E} | \ge (1 - \gamma) |Y^\star|\) (since \(|\widehat{T}| \ge (1 - \gamma) |T^\star|\)). We can bound the volume of \(\widehat{E}\) as
\begin{align*}
    \mathrm{vol}(\widehat{E}) &\le \exp \left( \frac{d \ln d}{\tauhat^2} \right) \cdot \mathrm{vol}(\widehat{B}_T) & \text{by (\ref{eq:ratio-transformed-volume})}\\
    &\le \exp \left( \frac{d \ln d}{\tauhat^2} \right) \cdot c_d \widehat{R}_T^d & c_d: \text{ vol.\ of $d$ dim.\ unit ball} \\
    &\le \exp \left( \frac{d \ln d}{\tauhat^2} \right) \cdot c_d \left((R^\star_T)^2 + \sigma^2 \gamma \left(1 + \frac{2(1 - \delta)}{\delta} \right)\right)^{d/2} & \text{by (\ref{eq:bthat-volume})} \\
    &\le \exp \left( \frac{d \ln d}{\tauhat^2} \right) \cdot c_d \left((R^\star)^2 + \frac{\sigma^2}{\gamma} \left(1 + \frac{2(1 - \delta)}{\delta} \right)\right)^{d/2} & R^\star \ge R^\star_T \text{ by (\ref{eq:volume-shrinking})} \\
    &\le \exp \left( \frac{d \ln d}{\tauhat^2} \right) \cdot c_d (R^\star)^d \left(1 + \frac{\sigma^2 }{\gamma (R^\star)^2} \left(1 + \frac{2(1 - \delta)}{\delta} \right)\right)^{d/2}\\
    &\le \exp \left( \frac{d \ln d}{\tauhat^2} \right) \cdot \exp \left(\frac{d}{2} \cdot \frac{\sigma^2 }{\gamma (R^\star)^2} \left(1 + \frac{2(1 - \delta)}{\delta} \right) \right) \cdot \mathrm{vol}(B^\star)  \\
    &\le \exp \left( \frac{d \ln d}{\tauhat^2} \right) \cdot \exp \left(\frac{\tauhat^2 (R')^2 }{2 \gamma (R^\star)^2} \left(1 + \frac{2(1 - \delta)}{\delta} \right) \right) \cdot \mathrm{vol}(B^\star)  & \text{since } \sigma^2 = \frac{\tauhat^2 (R')^2}{d} \\
    &\le \exp \left(\frac{d \ln d}{\tauhat^2} + \tauhat^2 \cdot \frac{2}{\gamma} \left(1 + \frac{2(1 - \delta)}{\delta} \right) \right) \cdot \mathrm{vol}(B^\star) &\text{since } R' \le 2 R^\star \\
    &\le \exp \left(d^{1/2} \left( \ln d + \frac{2}{\gamma} \left(1 + \frac{2(1 - \delta)}{\delta} \right) \right) \right) \cdot \mathrm{vol}(B^\star) &\text{setting } \tauhat^2 = d^{1/2}  \\
    &\le \exp \left(d^{1/2 + o(1)} (\gamma \delta)^{-1} \right) \cdot \mathrm{vol}(B^\star).
\end{align*}

Rephrasing the above bound in terms of \(\mathrm{vol}(\widehat{E})^{1/d}\) gives the stated bound. 

\end{proof}

\section{Greedy Algorithm and Unions of Sets}

In this section, we provide a greedy algorithm for constructing a union of sets from a given set system that achieves the desired coverage with a small volume compared to the minimum volume union of $k$ sets.







\begin{theorem}[Greedy algorithm]
    Let \(Y \subseteq \mathbb{R}^d\) be a set of \(n\) points, and \(\mathcal{C}\) be a set system over \(\mathbb{R}^d\) with bounded VC-dimension \(D\), \(k \in \mathbb{N}\),  \(\delta \in (0, 1)\) be a coverage level, and \(\gamma \in (0,1-\delta)\) be a slack parameter in coverage.  Let \(\mathcal{A}_\alpha(Y', \delta', \gamma')\) for approximation factor \(\alpha \ge 1\) be an algorithm, that given a set of points \(Y'\), and a coverage level \(\delta'\) and a slack $\gamma' \in (0,1)$, outputs a \(C \in \mathcal{C}\) such that \(|C \cap Y'| \ge (1 - \gamma') \delta |Y'|\), and \(\mathrm{vol}(C) \le \alpha \cdot \mathrm{vol}(C')\), where \(C'\) is the minimum volume set in \(\mathcal{C}\) such that \(|C' \cap Y'| \ge \delta |Y'|\).

    Then, in \(\mathrm{poly}(n, d, k)\) time and calls to \(\mathcal{A}_\alpha\), Algorithm \textsc{Greedy\_Density} outputs a set \(\widehat{C} \subseteq R^d\) such that \(\widehat{C}\) is a union of \(O(\delta k / \gamma)\) sets from \(\mathcal{C}\), and 
    \[\mathrm{vol}(\widehat{C}) \le O\left(\frac{\alpha \log (k/\gamma)}{\gamma}\right) \cdot \mathrm{vol}(C^\star),\]
    where \(C^\star\) is the minimum volume union of \(k\) sets from \(\mathcal{C}\) such that 
    \(|C^\star \cap Y| \ge (\delta + \gamma) n\).
    \label{thm:greedy}
\end{theorem}


\begin{figure}[htbp]
\begin{tcolorbox}

\begin{center}\textbf{Algorithm} \textsc{Greedy\_Density}\end{center}

\textbf{Input:} points \(Y \subseteq \mathbb{R}^d\) with \(|Y| = n\), coverage level \(\delta \in [0, 1]\), slack factor \(\gamma\), number of sets \(k \in \mathbb{N}\), black-box algorithm \(\mathcal{A}_\alpha\) (see description in Theorem \ref{thm:greedy}) \\
\textbf{Output:} confidence set $\widehat C \subseteq \R^d$

 \begin{enumerate}
     \item Set $t=0$ and $\widehat{C} = \varnothing$, \(Y_t = Y\). 
     \item While \(|\widehat{C} \cap Y| \le \delta n\):
     \begin{enumerate}
     \item For \(i \in \{0, \dots, \lceil \log_2 n \rceil \}\), let \(S^{(t)}_i = \mathcal{A}_\alpha(Y_t, \frac{2^i}{n}, \gamma')\) 
     \item Let 
     \[S^{(t)} = \max_{S^{(t)}_i} \frac{|S^{(t)}_i \cap Y_t|}{\vol(S^{(t)}_i)}\] 
     such that \(|S^{(t)}_i \cap Y_t| \ge (1 - \gamma') \frac{\gamma}{4k}\cdot n\) 

     \item Set \(\widehat{C} = \widehat{C} \cup S^{(t)}\), \(Y_{t + 1} = Y_{t} \setminus S^{(t)}\), \(t = t + 1\)
     
     \end{enumerate}

     \item Output \(\widehat{C}\)
 \end{enumerate}

\end{tcolorbox}
\label{fig:alggreedy}
\caption{Algorithm for finding the small  confidence set}
\end{figure}

\begin{proof}
Let $C^*_1, \dots, C^*_k \in \calF$ be the optimal $k$ sets in $\calF$ that covers at least a $\delta+\gamma$ fraction of points in \(Y\).
Let $C^* = C^*_1 \cup \dots \cup C^*_k$.

First we bound the number of iterations of the algorithm.  For each iteration of the algorithm, we have that \(|\widehat{C} \cap Y| \le \delta n\).  This means that \(|C^\star \cap Y_t| \ge \gamma n\), so there exists a \(C^\star_\ell \in C^\star\) such that \(|C^\star_\ell \cap Y_t| \ge \frac{\gamma n}{k}\).  Thus, step (b) of our algorithm will successfully find a set in \(S^{(t)} \in \mathcal{C}\) such that \(|S^{(t)} \cap Y_t| \ge (1 - \gamma') \frac{\gamma}{4k} \cdot n\).  Setting \(\gamma' = \frac{1}{2}\) we have that the algorithm must terminate in 
\[O \left(\frac{\delta}{\gamma } \cdot k \right) \]
iterations, and output a union of at most that many balls. 

Now we bound the volume of the output set.  We group the iterations into phases.  In particular, phase \(j\) consists of iterations in which \(|\widehat{C} \cap Y| \in [(1 - \frac{1}{2^{j}}) \delta n, (1 - \frac{1}{2^{j+ 1}}) \delta n )\).  \Cref{claim:marginal-volume-of-phase-bounded} tells us that the volume of sets added in each phase is 
\[O\left( \frac{\alpha}{\gamma} \right) \cdot \mathrm{vol}(C^\star) .\]

Since each iteration, the new picked set $S^{(t)}$ covers at least $(1-\gamma')\frac{\gamma}{4k}\cdot n$ points, 
there can be at most \(\log_2 (\delta k / \gamma) \le \log(k/\gamma)\) many phases. This means that the total volume of the sets output over all phases is bounded by 
\[O\left( \frac{\alpha}{\gamma} \cdot \log (k/\gamma) \right) \cdot \mathrm{vol}(C^\star).\]

\end{proof}

\begin{claim}[Marginal volume added in each phase is bounded]
In every phase $j$, the total marginal volume added is $$\le O(\alpha)\cdot\big( \vol(C^*_1)+\dots+\vol(C^*_k)\big).$$ 
\label{claim:marginal-volume-of-phase-bounded}
\end{claim}


\begin{proof}
Let \(t_{j}\) be the first iteration in phase \(j\), and \(t_{j + 1} - 1\) be defined as the last iteration in phase \(j\).  

We first show that in every iteration $t$, one of the sets $C^*_1, \dots, C^*_k$ is a feasible solution with good density. 
Note that at the beginning of each iteration, the confidence set covers at most $\delta$ fraction of points in $Y$.
Let $C^* = C^*_1 \cup \dots \cup C^*_k$.
Since $|C^\star \cap Y| \ge (\delta + \gamma) |Y| = (\delta + \gamma) n$, we have that \(|C^\star \cap Y_t| \ge \gamma n\) for any iteration $t$.

\Cref{claim:avging} tells us that there is a set \(C^\star_\ell\) that achieves individually high coverage and high density compared to \(C^\star\) over \(Y_t\).  That is, there is an \(\ell \in [k]\) such that 
\[|C^\star_\ell \cap Y_t| \ge \frac{\gamma n}{2k} \qquad \text{and} \qquad \frac{|C^\star_\ell \cap Y_t|}{\mathrm{vol}(C^\star_\ell)} \ge \frac{|C^\star \cap Y_t|}{2 \cdot \mathrm{vol}(C^\star)}.\]
Since we are in phase \(j\), we have that the number of points that we have covered thus far is \(|Y \setminus Y_t| \le (1 - \frac{1}{2^{j + 1}}) \delta n \), so \(|C^\star \cap Y_t| \ge \frac{\delta n}{2^{j + 1}}\), so 
\[\frac{|C^\star_\ell \cap Y_t|}{\mathrm{vol}(C^\star_\ell)} \ge \frac{|C^\star \cap Y_t|}{2 \cdot \mathrm{vol}(C^\star)} \ge \frac{\delta n}{2^{j + 2} \cdot \mathrm{vol}(C^\star)} .\]
This means that, by \Cref{claim:find-approx-density-maximizer}, will choose a set \(C^{(t)} \in \mathcal{C}\) such that 
\[|C^{(t)} \cap Y_t| \ge (1 - \gamma') \frac{\gamma n}{4k} \qquad \text{and} \qquad \frac{|C^{(t)} \cap Y_t|}{\mathrm{vol}(C^{(t)})} \ge \frac{|C^\star_\ell \cap Y_t|}{\alpha \cdot \mathrm{vol}(C^\star_\ell)} \ge \frac{\delta n}{\alpha \cdot 2^{j + 2} \cdot \mathrm{vol}(C^\star)}.\]

This implies that 
\begin{align*}
    \frac{|C^{(t)} \cap Y_t|}{\mathrm{vol}(C^{(t)})} &\ge \frac{\delta n}{\alpha \cdot 2^{j + 2} \cdot \mathrm{vol}(C^\star)} \\
    \frac{|C^{(t)} \cap Y_t| \cdot \alpha \cdot 2^{j + 2} \cdot \mathrm{vol}(C^\star)}{\delta n} &\ge \mathrm{vol}(C^{(t)}) .
\end{align*}
Now we can aggregate over the iterations in phase \(j\).  We first consider all but the last iteration, that is \(t\) such that \(t_j \le t \le t_{j + 1} - 2\).  (It can be the case that \(t_{j + 1 } - 2 < t_j\), in which case there are no \(t\) in this category and the bound holds trivially.)
\begin{align*}
    \left(\sum_{t = t_j}^{t_{j + 1} - 2} |C^{(t)} \cap Y_t| \right) \cdot \frac{\alpha \cdot 2^{j + 2} \cdot \mathrm{vol}(C^\star)}{\delta n} &\ge \sum_{t = t_j}^{t_{j + 1} - 2} \mathrm{vol}(C^{(t)}) 
\end{align*}
Now, since for all iterations \(t\) such that \(t_j \le t \le t_{j + 1} - 2\), we have that \(t + 1\) is still in phase \(j\), we know that the chosen sets could not have covered more than \(\frac{1}{2^{j}} n\) points.
\begin{align*}
    \frac{\delta n}{2^j} \cdot \frac{\alpha \cdot 2^{j + 2} \cdot \mathrm{vol}(C^\star)}{\delta n} &\ge \sum_{t = t_j}^{t_{j + 1} - 2} \mathrm{vol}(C^{(t)}) \\
    O(\alpha) \cdot \mathrm{vol}(C^\star) &\ge \sum_{t = t_j}^{t_{j + 1} - 2} \mathrm{vol}(C^{(t)}).
\end{align*}

Finally, for the iteration \(t = t_{j + 1} - 1\), we have that \(|C^{(t)} \cap Y_t| \le n\).  So we get that 
\begin{align*}
    \frac{|C^{(t)} \cap Y_t|}{\mathrm{vol}(C^{(t)})} &\ge \frac{1}{2\alpha} \cdot \frac{|C^\star_\ell \cap Y_t|}{\mathrm{vol}(C^{\star}_\ell)} &\text{by \Cref{claim:find-approx-density-maximizer}}\\
    \frac{|C^{(t)} \cap Y_t|}{\mathrm{vol}(C^{(t)})} &\ge \frac{1}{4\alpha} \cdot \frac{|C^\star \cap Y_t|}{\mathrm{vol}(C^{\star})} &\text{by \Cref{claim:avging}}\\
    \frac{n}{\mathrm{vol}(C^{(t)})} &\ge \frac{1}{4\alpha} \cdot \frac{\gamma n}{\mathrm{vol}(C^{\star})} \\
    \frac{4 \alpha }{\gamma} \mathrm{vol}(C^\star) &\ge \mathrm{vol}(C^{(t)}).
\end{align*}

Thus in total, over all days in phase \(j\), we have that the total volume of sets chosen by our algorithm is at most 
\[O \left(\frac{\alpha}{\gamma} \right) \cdot \mathrm{vol}(C^\star).\]

\end{proof}

\begin{claim}\label{claim:avging}
Given $a_1, a_2, \dots, a_k, b_1, b_2, \dots, b_k \geq 0$ such that $\sum_{i=1}^k a_i\ge \beta$, and $\sum_{i=1}^k a_i / \sum_{i=1}^k b_i \ge \gamma$. Then there exists $i \in [k]$, satisfying $a_i \ge \beta/(2k)$ and $a_i/b_i \ge \gamma/2$.
\end{claim}
\begin{proof}
Let $\beta' = \sum_{i=1}^k a_i$.
Suppose there is no $a_i,b_i$ such that $a_i \ge \beta'/(2k)$ and $a_i/b_i \ge \gamma/2$. Let $B=\{i \in [k]: a_i \ge \beta'/(2k)\}$. Then, we have
$$\sum_{i \in B} a_i = \sum_{i \in [k]} a_i - \sum_{i \in [k] \setminus B} a_i > \sum_{i \in [k]} a_i - k \cdot \frac{\beta'}{2k} = \frac{\beta'}{2}.$$ 
Moreover for all $i \in B$ we have $a_i/b_i <  \gamma/2$ i.e. 
$$\sum_{i \in B} b_i > \frac{2}{\gamma} \sum_{i \in B} a_i > \frac{2}{\gamma} \cdot \frac{\beta'}{2} = \frac{\beta'}{\gamma}.$$
But this contradicts $\sum_{i\in [k]} b_i \leq \sum_{i=1}^k a_i/\gamma = \beta'/\gamma$, which completes the proof. 
\end{proof}

\begin{claim}[Each iteration finds approximate density maximizer]
In each iteration, our algorithm chooses a set \(S^{(t)}\) such that 
\[\frac{|S^{(t)} \cap Y_t|}{\mathrm{vol}(S^{(t)})} \ge \frac{1}{2\alpha} \cdot \frac{|C^{(t)} \cap Y_t|}{\mathrm{vol}(C^{(t)})},\]
and  
\[(1 - \gamma') \frac{\gamma}{4k} \cdot n \le |S^{(t)} \cap Y_t|,\]
where \(C^{(t)}\) is the maximum density set such that \(\frac{\gamma}{2k} n \le |C^{(t)} \cap Y_t|.\)
\label{claim:find-approx-density-maximizer}
\end{claim}

\begin{proof}
    Fix an iteration \(t\), and let \(C^{(t)}\) be the set in \(\mathcal{C}\) such that \(\frac{\gamma}{2k} n \le |C^{(t)} \cap Y_t| \), that maximizes 
    \[\frac{|C^{(t)} \cap Y_t|}{\mathrm{vol}(C^{(t)})}.\]

    Let \(i^* \in \{0, \dots, \lceil \log_2 n \rceil \}\) be the value that satisfies 
    \[2^{i^\star} \le |C^{(t)} \cap Y_t| \le 2 \cdot 2^{i^\star}. \]
    Since \(C^{(t)}\) achieves coverage at least \(\frac{2^{i^\star}}{n}\), we have that the set \(S_{i^\star}^{(t)}\) chosen in step (a) of the algorithm must have
    \[\mathrm{vol}(S_{i^\star}^{(t)}) \le \alpha \cdot \mathrm{vol}(C^{(t)}) \qquad \text{and} \qquad |S_{i^\star}^{(t)} \cap Y_t| \ge (1 - \gamma') \frac{2^{i^\star}}{n}.\]
    Thus, we have that 
    \[\frac{|S_{i^\star}^{(t)} \cap Y_t|}{\mathrm{vol}(S_{i^\star}^{(t)})} \ge \frac{1}{2\alpha} \cdot \frac{|C^{(t)} \cap Y_t|}{\mathrm{vol}(C^{(t)})},\]
    and that 
    \[(1 - \gamma') \frac{\gamma}{4k} \cdot n \le |S_{i^\star}^{(t)} \cap Y_t|.\]
\end{proof}




\begin{corollary}[Union of Balls]
    We give an algorithm that, given a set of points \(Y \subseteq \mathbb{R}^d, |Y| = n\), a coverage fraction $\delta \in (0,1)$, a slack parameter in coverage $\gamma \in (0,1-\delta)$ and $k \in \N$, can find a set \(\widehat{C}\) such that 
    \[\mathrm{vol}(\widehat{C})  \le \exp \left(O_\delta \left(d^{1/2 + o(1)} \right)\right) \cdot O\left( \frac{1}{\gamma} \cdot \log (k / \gamma) \right) \cdot  \mathrm{vol}(C^\star)  ,\]
    and \(\widehat{C}\) is a union of \(O(\frac{\delta k}{\gamma})\) ellipsoids, where \(C^\star\) is the minimum volume union of \(k\) balls that covers at least \((\delta + \gamma)\) fraction of the points in \(Y\).
    \label{cor:union-of-ellipsoids}
\end{corollary}

\begin{corollary}[Proper Learning Union of balls]
    Let $\delta \in (0,1), \gamma \in (0,1), k \in \mathbb{N}$ be any constants. There is a polynomial time algorithm that for target coverage \(\delta \in (0, 1)\) and coverage slack $\gamma\in(0,1)$ when given $n=\Omega(kd^2/\gamma^2)$ samples drawn i.i.d. from an arbitrary distribution $\mathcal{D}$,  finds with high probability a set $S \subset \mathbb{R}^d$ that is a union of balls, and is $\Gamma=\left(1 + O_{\gamma, \delta} \big(\log \log d / \log d\big)\right)$ competitive; more precisely, it satisfies
    $\mathbf{P}_{y \sim \mathcal{D}} \left[ y \in S \right] \ge \delta$,
    and 
    \[\vol(S)^{1/d} \le  \mathrm{vol}(C^\star_k)^{1/d} \Big(1+ O_{k,\delta} \left(\log \log d / \log d \right)\Big) \cdot \left( \frac{O(\log(k/\gamma))}{\gamma} \right)^{1/d}  \]
    where \(C^\star_k\) is the minimum volume union of $k$-balls that achieves at least \(\delta + \gamma + O(\sqrt{kd^2/n})\) coverage over \(\mathcal{D}\).
    \label{cor:learning_union-of-balls}
\end{corollary}

\section{Hardness of Proper Learning}

\subsection{NP-hardness of Proper Learning}
In this section, we show the approximation hardness of finding the smallest volume ball that achieves the required coverage. 

\begin{theorem}[NP-hardness of Proper Learning]\label{thm:properhard}
    For any small constant $\varepsilon>0$ and there is some constant $\delta \geq 1/4$, unless $P=NP$, there is no algorithm that given a set of points $Y \subseteq \R^d$ runs in polynomial time in $|Y|, d$ that finds a ball $\Bhat = B(\chat,\Rhat)$ that contains at least a $\delta$ fraction of points in $Y$ with radius $\Rhat\le \big(1+d^{-\varepsilon}\big) R^*$, where $R^*$ is the radius of the smallest ball $B^*$ that contains at least a $\delta$ fraction points in $Y$. 
    \anote{Edited the statement to give a lower bound on $\delta$.}
\end{theorem}


\anote{Added 5/12 (to address Reviewer):} We remark that the above hardness result rules out polynomial time proper learning algorithms that work for arbitrary distributions--- this is called ``the distribution-free setting'' in PAC learning (as opposed to the distribution-specific setting). More formally, one can consider the following reduction from the worst-case problem of finding the minimum volume ball containing a $\delta$ fraction of an adversarially chosen point set $Y \subset \R^d$, to our problem. We construct a distribution $D_Y$ that is the uniform distribution over $Y$. By the construction of $Y$, finding a set that contains probability mass $\delta$ over $D_Y$ corresponds exactly to finding a set that contains at least a $\delta$ fraction of $Y$. We can simulate running our algorithm on $D_Y$ by sampling from $D_Y$ and providing these samples to the algorithm. Thus our problem is only harder than the worst-case problem, because it only gets sample access to $D_Y$, rather than access to $Y$ itself.

The above theorem involves a reduction from the maximum clique problem to the smallest $k$-enclosing ball due to \cite{ballsnphard}.
We now proceed to the proof of Theorem~\ref{thm:properhard}.

\begin{proof}
    \cite{ballsnphard} proves the strong NP-hardness of the smallest $k$-enclosing ball in Euclidean space through a reduction from the $k$ clique problem on regular graphs. 
    We utilize this construction in our reduction as follows. 
    Let $G$ be a $\Delta$-regular graph with $n$ vertices and $m$ edges. We construct a set of $n$ points $Y \subset \R^{d}$ with dimension $d = m^{3/\epsilon}$ for a small constant $\epsilon \in (0,1)$. Each point $y \in Y$ corresponds to a vertex $v$ in graph $G$. The first $m$ coordinates of $y$ are the $m$-dimensional row for vertex $v$ in the incident matrix of $G$. Then, for $i = 1,2,\dots, m$, we have $y_i = 1$ if the edge $i$ connects vertex $v$; otherwise $y_i = 0$. The rest $d- m$ coordinates of $y$ have value $0$. Note that this reduction takes time $O(nd)$, which is polynomial in $m$ when $\varepsilon$ is a small constant in $(0,1)$.

    We now show that this construction of $Y$ provides the hardness of proper learning.
    Since the graph $G$ is $\Delta$ regular, we have the distance between two points in $Y$ is $\sqrt{2\Delta-2}$ if the corresponding vertices are connected in $G$; otherwise, the distance is $\sqrt{2\Delta}$.
    By Lemma 3 in~\cite{ballsnphard}, we have if $G$ has a $k$ clique, then the smallest ball that contains at least $k$ points in $Y$ has a radius $R^* \leq \sqrt{A_k}$ where $A_k = (\Delta-1)(1-1/k)$; otherwise, the smallest ball that contains $k$ points in $Y$ has a radius at least $\sqrt{A_k + 2/k^2}$. Note that $\sqrt{A_k + 2/k^2} \geq \sqrt{A_k} (1+1/n^3) \geq \sqrt{A_k} (1+1/m^3)$. The $k$-clique problem is NP-hard even for regular graphs~\cite{fleischner2010maximum, feige2003regularVC}. 
    \cite{fleischner2010maximum} show that it is NP-hard to find the maximum independent set for planar three regular graphs. Note that for three regular graphs, the maximum independent set has a size of at least $n/4$. The complement of a regular graph is still a regular graph and the maximum independent set corresponds to the maximum clique in the complement graph. 
    Thus, for some constant $\delta \geq 1/4$, given a regular graph $G$ on $n$ vertices that contains $k = \delta n$ clique, it is NP-hard to find a $k$ clique in $G$. 
    
    Since we pick $d= m^{3/\varepsilon}$, we have it is NP-hard to approximate the radius of the smallest ball that contains a $\delta$ fraction of points in $Y$ within a factor of 
    $$
    1+\frac{1}{m^3} = 1+\frac{1}{d^{\varepsilon}}.
    $$
    Thus, it is NP-hard to find a ball that contains a $\delta$ fraction of points in $Y$ with a radius at most $1+\frac{1}{d^\varepsilon}$ times the radius $R^*$ of the minimum ball that contains a $\delta$ fraction of points in $Y$.
\end{proof}

\subsection{Computational Intractability even with Slack in Coverage}

We give different reduction to provide evidence of strong computational intractability even when we are allowed to violate the coverage by a constant factor. Our hardness result is assuming the Small Set Expansion (SSE) hypothesis of ~\cite{raghavendra2010graph}, which is closely related to the Unique Games Conjecture~\cite{khot2003UGC}.

\begin{conjecture}[SSE hypothesis of~\cite{raghavendra2010graph}, see e.g., Theorem IV.5 of \cite{RST2012}] \label{conj:sse} For any constant $\eta \in (0,1/2)$, there is a constant $\tau\in (0,1)$ such that there is no polynomial time algorithm to distinguish between the following two cases given a graph $G=(V,E)$ on $n$ vertices with degree $D$:
\begin{itemize}
\item YES: Some subset $S\subseteq V$ with $|S| = \tau n$ satisfies that the induced subgraph on $S$ is dense i.e., the number of edges going out of $S$ is $|E(S,V\setminus S)| \le \eta D |S|$ edges.  
\item NO: Any set $S\subseteq V$ with $\eta \tau n \le |S| \le 2\tau n$ has most of the edges incident on it going outside i.e., $|E(S,V\setminus S)| \ge (1-\eta) |S| D$.
\anote{A lower bound on the size of $S$ can be imposed. }
\end{itemize}
\end{conjecture}

We prove the following theorem. 
\begin{theorem}[Computational Intractability with Slack in Coverage]\label{thm:properhard:slack}
    For any constant $\gamma>0$, there exists a constant $\delta \in (0,1)$, 
    such that assuming the SSE hypothesis for any constant $\varepsilon>0$ there is no algorithm that given a set of points $Y \subseteq \R^d$ runs in polynomial time and finds a ball $\Bhat = B(\chat,\Rhat)$ that contains at least a $\gamma \delta$ fraction of points in $Y$ with radius $\Rhat\le \big(1+d^{-\varepsilon}\big) R^\star$, where $R^\star$ is the radius of the smallest ball $B^*$ that contains at least a $\delta$ fraction points in $Y$. 
    \anote{Edited the statement to give a lower bound on $\delta$.}
\end{theorem}
\begin{proof}
Set $\eta \coloneqq \min\{\tfrac{\gamma^2}{9}, \tfrac{1}{16}\}$. Given a $D$-regular graph $G=(V,E)$ of the SSE problem in Conjecture~\ref{conj:sse} on $n$ vertices with parameter $\eta$, we first construct an instance for ball coverage in $\R^m$ where $m = |E|$ as follows. Every edge corresponds to a coordinate, and every vertex corresponds to a point in $\R^m$ that corresponds to its edge incidences i.e., $u_i (e) = 1$ if vertex $i$ is an endpoint of edge $e$, and $0$ otherwise. Note that $\|u_i - u_j\|^2 = 2D-2$ if $(i,j) \in E$ and $\|u_i - u_j\|^2=2D$ when $(i,j) \notin E$. 
Let $k=\tau n$ where $\tau \in (0,1)$ is the parameter in Conjecture~\ref{conj:sse}. Let $R^\star=\sqrt{D - \frac{D}{k}}$, and $R' = \sqrt{D - \frac{D}{2k}}$. We first prove that it is SSE-hard to distinguish between the following two cases:
\begin{itemize}
    \item (YES) case when there exists a ball of radius $R^\star$ that contains at least $\delta n$ points,
    \item (NO) case when every ball of radius at most $R'$ has at most $\gamma \delta n$ points inside it. 
\end{itemize}

\noindent {\em Completeness argument (YES case):} In the YES case of SSE hypothesis, there exists a subset $S \subset V$ of size $k$ with $|E(S,S)| \ge \frac{k D (1-\eta)}{2}$. Consider a candidate center 
$$c^\star = \begin{cases} \frac{2}{k} & \text{ if } (i,j) \in E(S,S)\\ 0 &\text{ otherwise} 
\end{cases}.$$

Let $S^\star= \{i \in S: \text{deg}_{S}(i) \ge (1-\sqrt{\eta}) D\}$. Since $|E(S,S^c)|\le \eta k D$, by Markov's inequality, $|S^\star| \ge (1-\sqrt{\eta}) |S| \geq 3/4\cdot |S| = \delta n$ where $\delta = 3\tau/4$. We have $S^\star \subset \Ball(c^\star, R^\star)$ since
\begin{align*}
\forall i \in S^\star, ~~\| u_i - c^\star\|_2^2 = & \sum_{e \in E} \Big(u_i(e) - \frac{2}{k} \Big)^2 = (D- \text{deg}_{S}(i)) + \text{deg}_{S}(i) \Big(1-\frac{2}{k} \Big)^2 \\
&+ (|E(S,S)|- \text{deg}_S(i))\Big( \frac{2}{k}\Big)^2 \\
= & D - \text{deg}_{S}(i) \big(\frac{4}{k}-\frac{4}{k^2}\big) + (|E(S,S)|-\text{deg}_S(i)) \cdot \big(\frac{4}{k^2}\big) \\
\le& D +  \frac{kd}{2} \cdot \frac{4}{k^2} - \frac{4(1 -\sqrt{\eta})D}{k} 
\le D - (1-2\sqrt{\eta})\frac{2D}{k} \\
\le & (R^\star)^2, \text{ since } \eta\le 1/16. 
\end{align*}


\noindent {\em Soundness argument (NO case):} Suppose there exists a subset of points corresponding to vertices $S' \subset V$ with $|S'| =k'$ for $3\gamma k/4  \leq k' \leq k$ that are contained in a ball of radius $R'$ around a point $c' \in \R^n$. Note that $|S'| \ge \eta \tau n$ and $|S'| =k' \geq 3\gamma k/4 = \gamma \delta n$.
We have
\begin{align*}
(R')^2 &=  \max_{i \in S'} \| u_{i} - c' \|^2 \ge \text{Avg}_{i \in S'} \| u_{i} - c' \|^2 \ge  \text{Avg}_{i \in S'} \Big\| u_{i} - \text{Avg}_{i \in S'} u_i \Big\|^2\\
&= \frac{1}{2}\text{Avg}_{i,j \in S'} \| u_i - u_j\|^2 = D- \frac{2|E(S',S')|}{|S'|(|S'|-1)}.  \end{align*}
Since $3\gamma k/4 \le |S'| \le k$, we have
\begin{align*}
D - \frac{2|E(S',S')|}{k'(k'-1)} &\le (R')^2 \le D- \frac{D}{2k}\\
|E(S',S')| &\ge \frac{D}{4k} \cdot k'(k'-1) \ge \frac{3\gamma \big( 3\gamma k/4 - 1\big)D}{16} \ge \frac{\gamma^2}{8} kd > \eta D |S'|,
\end{align*}
since $\gamma^2 > 8 \eta$. This contradicts the NO case of the SSE conjecture. This finishes the soundness argument. Thus, we have established the hardness of approximating the minimum radius ball containing $\delta n$ points, even when we are allowed an arbitrary constant slackness factor $\gamma>0$ in the coverage.  It is hard to approximate the radius within a factor of 
$$
\frac{R'}{R^*} = \sqrt{\frac{D-D/(2k)}{D-D/k}} \geq \sqrt{1 + \frac{1}{2k}} \geq 1 + \frac{1}{4m}.
$$



Now as in Theorem~\ref{thm:properhard} we can pad the instance with dummy coordinates to make it a $d$ dimensional instance with $d=(4m)^{1/\varepsilon}$ to get the desired $1+d^{-\varepsilon}$ inapproximability factor in radius. 


\end{proof}

\section{Application to Conformal Prediction}\label{sec:conformal}
As an immediate application 
of our result, we obtain an algorithm and guarantee for conformal prediction with approximate volume optimality in the high-dimensional setting.  


Conformal prediction is the statistical problem of finding prediction intervals.  That is, given training examples \(Y_1, \dots, Y_n\) lying in some space \(\mathcal{Y}\), and a target miscoverage rate \(\alpha > 0\), our goal is to output a set \(C\), such that for an unknown test example \(Y_{n + 1} \in \mathcal{Y}\),
\begin{equation}
    \Pr[Y_{n + 1} \in C] \ge 1 - \alpha, \label{eq:conformal-coverage}
\end{equation}
assuming \emph{exchangeability} of the training examples and the test example, i.e., 
\[\Pr[ Y_1 = y_1, \dots, Y_{n + 1} = y_{n + 1}] = \Pr[ Y_1 = y_{\pi(1)}, \dots, Y_{n + 1} = y_{\pi(n + 1)}],\]
for all \(y_1, \dots, y_{n + 1} \in \mathcal{Y}\) and permutations \(\pi\) over \(\{1, \dots, n + 1\}\).\footnote{We refer to this setting as the \emph{unsupervised} setting. In the more general \emph{supervised} setting, our training examples are feature-label pairs \((X_i, Y_i) \in \mathbb{X} \times \mathcal{Y}\), and our task is given \(X_{n + 1}\) to output a set \(C(X_{n + 1})\) that contains the (unknown) \(Y_{n + 1}\) with probability \(\ge 1 - \alpha\), assuming exchangeability of the training and test examples.  Strategies for the unsupervised setting often translate to the supervised setting where a regression model is trained to predict \(Y_i\) from \(X_i\), by applying the unsupervised conformal inference procedure on the \emph{residuals} (error between the true \(Y_i\) and predicted \(Y_i\)), or by considering the conditional distribution $Y_i \mid X_i$.  See e.g., \cite{angelopoulos2024theoretical, ourwork2024} for details.}  
We refer the reader to the book of \cite{angelopoulos2024theoretical} for a thorough introduction and treatment of conformal prediction.

There are many ways to construct conformal predictors that achieve (\ref{eq:conformal-coverage}).  For example, a trivial predictor can simply output the whole space \(\mathcal{Y}\), and achieve coverage \(1\). %
The common measure of efficiency that is used to compare conformal predictors is volume. That is, if \(\mathcal{Y} = \mathbb{R}^d\), then we would like a conformal predictor that minimizes \(\mathbb{E} \left[\mathrm{vol}(C) \right]\), where \(\mathrm{vol}(C)\) is the Lebesgue measure of the set \(C \subseteq \mathbb{R}^d\).  Typically in the literature, a conformal predictor must \emph{provably} achieve coverage (\ref{eq:conformal-coverage}) under only the weak condition of exchangeability, and has efficiency (volume) that is empirically validated on datasets. 

Much of the work in conformal prediction has focused on the one-dimensional setting where \(\mathcal{Y} = \mathbb{R}\).  In these settings, it is clear that any natural set \(S \subseteq \mathbb{R}\) is a union of intervals.  \vnote{cite some papers here}  However, a recent line of work has explored the setting where \(\mathcal{Y} = \mathbb{R}^d\) is higher dimensional.  In this setting, a priori it is not even clear what form the output set \(S \subseteq \mathbb{R}^d\) should take.  
\cite{wang2023probabilistic} and \cite{zheng2025generativeconformalpredictionvectorized} provide methods to tackle the problem by outputting sets \(S\) that are unions of balls.  They provide theoretical guarantees that their methods achieve coverage, and validate the efficiency of their methods empirically on real and synthetic data.  

One may hope to design a conformal predictor that is \emph{provably} volume optimal subject to  achieving coverage (\ref{eq:conformal-coverage}).  However, even when the points \(Y_i\) are drawn i.i.d.\ from some arbitrary distribution \(\mathcal{D}\) over \(\mathbb{R}^d\), the problem of finding the minimum volume set \(C \subseteq \mathbb{R}^d\) such that 
\[\Pr[Y_{n + 1} \in C] = \Pr_{y \sim \mathcal{D}}[y \in C] \ge 1-\alpha\]
is statistically intractable even for $d=1$.  Thus any provable guarantee for volume optimality must restrict the problem in some way.  \cite{ourwork2024} observes that the problem becomes statistically tractable (for i.i.d. samples) when we restrict our set \(C\) to come from some class of bounded VC-dimension \(\mathcal{C}\).  That is, we compete with 
\[\min_{C \in \mathcal{C}} \mathrm{vol}(C) \qquad \text{s.t.} \quad \Pr_{y \sim \mathcal{D}}[y \in C] \ge 1 - \alpha,\]
which they term \emph{\(\mathcal{C}\)-restricted volume optimality}.
In fact, in this setting, the problem essentially reduces to finding the minimum volume set in \(\mathcal{C}\) that achieves coverage \(1 - \alpha\) assuming the samples are i.i.d., which is precisely the problem that we solve (approximately) in this work.  A conformal predictor must additionally satisfy (\ref{eq:conformal-coverage}) when the samples are exchangeable, but this can be done with a standard ``conformalizing" step.

This gives us the following sample theorem when the class \(\mathcal{C}\) is the set of Euclidean balls.  (Equivalent statements go through for the other settings we consider in this work: properly learning Euclidean balls, and properly/improperly competing with unions of balls.  See \Cref{remark:other-C}.)

\highdimensionalconformalprediction*


\begin{proof} Without loss of generality, we assume that $n$ is even.
    \Cref{thm:intro:ellipsoid} gives an algorithm that satisfies case (b).  That is, given samples \(Y_1, \dots, Y_{n/2}\) drawn i.i.d.\ from \(\mathcal{D}\), for \(n =  \Omega(d^2/\gamma^2)\), it outputs a set 
    \[\mathrm{vol}(\widehat{C}) \le \left(1 + O_{\gamma, \delta} \big(d^{-1/3 + o(1)}\big)\right) \mathrm{vol}(C^\star),\]
    where 
    \[C^\star = \argmin_{C \in \mathcal{C}} \mathrm{vol}(C) \qquad \text{s.t.} \quad \Pr_{y \sim \mathcal{D}}[y \in C] \ge 1 - \alpha + \gamma.\]
    Since \(\Pr_{y \sim \mathcal{D}}[y \in C] = \Pr[Y_{n + 1} \in C]\), this satisfies case (b). 

    Now, to construct a conformal predictor, it suffices to give a \emph{conformity score}.\footnote{See Definition 3.1 in \cite{angelopoulos2024theoretical}.}  We can do this by constructing a \emph{nested set system}, following the strategy in \cite{gupta2022nested}.\footnote{See Assumption 2.4 in \cite{ourwork2024} for a formal statement.}  Let \(\mu \in \mathbb{R}^d\) be the center of \(\widehat{C}\).  We define the natural scaling of \(\widehat{C}\) by a scalar \(\lambda \ge 0\) as 
    \begin{equation}
        \lambda \widehat{C} = \{ \lambda (x - \mu) + \mu ~|~ x \in \widehat{C}\}. \label{eq:natural-scaling}
    \end{equation}
    For \(0 < \tau < 1\), define 
    \begin{equation}
        \lambda_\tau = \argmin_{\lambda \ge 0} \text{ s.t. } \left| \{Y_{n/2+1}, \dots, Y_{n + 1}\} \cap \lambda \widehat{C} \right| \ge \tau n/2 . \label{eq:scaling-based-on-tau}
    \end{equation}
    That is, \(\lambda_\tau\) is the smallest scaling of \(\widehat{C}\) that achieves coverage \(\tau\) over the training samples.  Our nested set system will consist of \(\lambda_\tau \widehat{C}\) for values \(\tau\) from a grid over the interval \([0, 1]\).  It will also contain \(\widehat{C}\) itself.  Since all sets in this system are scalings of the same convex set centered at the same point, they are indeed nested.  

    We can use this set system to construct a conformity score, see e.g. equation (10) and Assumption 2.4 in \cite{ourwork2024}, which implies a conformal predictor via split conformal prediction, see e.g. Algorithm 3.6 in \cite{angelopoulos2024theoretical}.  These have the property that they will only output sets \(C\) from the input nested set system.  Thus, we have that 
    \begin{enumerate}[(a)]
        \item when \(Y_1, \dots, Y_{n + 1}\) are exchangeable, then we achieve coverage
        \[\Pr [Y_{n + 1} \in \widehat{C} \ge 1 - \alpha].\]
        \item if \(Y_1, \dots, Y_{n + 1}\) are drawn i.i.d.\ from some (unknown) distribution \(\mathcal{D}\), and \(n =  \Omega(d^2/\gamma^2)\), then since \(\widehat{C}\) is in our nested set system, and \Cref{thm:intro:ellipsoid} guarantees that \(\widehat{C}\) achieves coverage \(\ge 1 - \alpha\), the conformal predictor will not output any set in the system larger than \(\widehat{C}\). 
    \end{enumerate}
\end{proof}

\begin{remark}\label{remark:other-C}
    The main step that is necessary to ``conformalize" our approximation result is to create a nested set system.  That is, given the output \(\widehat{C}\) of our approximation algorithm, which is a set that provably achieves coverage \(\ge 1 - \alpha\) over \(Y_{n + 1}\) if the samples \(Y_1, \dots, Y_n\) were drawn i.i.d.\ from some (unknown) distribution \(\mathcal{D}\), we must create a family of nested sets \(\widehat{\mathcal{C}}\), such that:
    \begin{enumerate}[(a)]
        \item \(\widehat{C} \in \widehat{\mathcal{C}}\),
        \item For a wide range of coverage levels \(\tau\) in some set \(T\),\footnote{Technically speaking, to achieve coverage (\ref{eq:conformal-coverage}), we only need that \(T\) contains some coverage level close to 1. However, for the conformal predictor to be efficient, it is good to think of \(T\) as being a grid over \([0, 1]\). } we have that there exists a set \(C_\tau \in \widehat{\mathcal{C}}\) such that 
        \[\left| \{Y_{n/2+1}, \dots, Y_{n + 1}\} \cap \lambda C_\tau \right| \ge \tau n ,\]
        and ideally this quantity is close to \(\tau n/2\).
    \end{enumerate}
    All of the forms of \(\widehat{C}\) that we work with in this work have natural notions of scaling (equivalent of \Cref{eq:natural-scaling}) that result in nested sets.  That is, for balls we can use the same scaling as \Cref{eq:natural-scaling}, and for unions of balls/ellipsoids, we can scale each set in the union individually.  Thus there is a natural way to create this set system, and the rest of the argument from \Cref{eq:scaling-based-on-tau} onward goes through unchanged to get the equivalent guarantees.
\end{remark}



\section{Acknowledgements}
This research project was supported by NSF-funded Institute for Data, Econometrics, Algorithms and Learning (IDEAL) through the grants NSF ECCS-2216970 and ECCS-2216912. The research started as part of the IDEAL special program on Reliable and Robust Data Science.  Chao Gao was also supported in part by NSF Grants DMS-2310769, NSF CAREER Award DMS-1847590, and an Alfred Sloan fellowship. Vaidehi Srinivas was supported by the Northwestern Presidential Fellowship.  We gratefully acknowledge the support of the NSF-Simons AI-Institute for the Sky (SkAI) via grants NSF AST-2421845 and Simons Foundation MPS-AI-00010513, and the support of the NSF-Simons National Institute for Theory and Mathematics in Biology (NITMB) via grants NSF DMS-2235451 and Simons Foundation MP-TMPS-00005320. We also thank an anonymous reviewer for helpful pointers to previous work involving coresets.

\bibliographystyle{plainnat}
\bibliography{ref}

\end{document}